\documentclass[submission,copyright,creativecommons]{eptcs}
\usepackage{amssymb,latexsym,leqno,breakurl}
\usepackage{theorem} %allows adjustment of fonts for theorems
\usepackage{amsmath} %for conditional def's

\usepackage{hyperref} 

\newsavebox{\blah}
\savebox{\blah}{\scriptsize $\{$}
\newsavebox{\blahb}
\savebox{\blahb}{\scriptsize $\}$}

\newcommand{\ie}{{\em i.e.}}
\newcommand{\eg}{{\em e.g.}}

\newcommand{\mathify}[1]{\ifmmode{#1}\else\mbox{$#1$}\fi}
\newcommand{\blackslug}{\rule{7pt}{7pt}}
\newcommand{\Rule}[2]{          % Unnumbered rule
  \begin{array}{c}
  #1 \\\hline
  #2
  \end{array}}

\newcommand{\mletml}[3]{\MLET\;#1=#2\;\MIN\;#3}

\newcommand{\ifthenelse}[3]{\IF\;#1\;\THEN\;#2\;\ELSE\;#3}
\newcommand{\mifthenelse}[3]{\MIF\;#1\;\MTHEN\;#2\;\MELSE\;#3}
\newcommand{\enable}[2]{\ENABLE\;#1\;\IN\;#2}
\newcommand{\chk}[2]{\CHK\;#1\;\FOR\;#2}
\newcommand{\test}[3]{\TEST\;#1\;\THEN\;#2\;\ELSE\;#3}
\newcommand{\signs}[2]{\SIGNS\;#1\;#2}

\newcommand{\vmlet}[2]{\begin{array}[t]{l} % Dave dropped six \!
                            \MLET\;{#1}\;\MIN\\
                            {#2}
                            \end{array}}

\newcommand{\BOOL}{\mbox{\texttt{bool}}}
\newcommand{\ELSE}{\mbox{\texttt{else}}}

\newcommand{\IF}{\mbox{\texttt{if}}}
\newcommand{\FOR}{\mbox{\texttt{for}}}
\newcommand{\IN}{\mbox{\texttt{in}}}
\newcommand{\MIN}{\mbox{\textbf{in}}}

\newcommand{\LET}{\mbox{\texttt{let}}}
\newcommand{\MLET}{\mbox{\textbf{let}}}

\newcommand{\REC}{{\texttt{rec}}}

\newcommand{\THEN}{{\texttt{then}}}
\newcommand{\MTHEN}{{\textbf{then}}}
\newcommand{\MELSE}{{\textbf{else}}}
\newcommand{\MIF}{{\textbf{if}}}
\newcommand{\TRUE}{{\texttt{true}}}
\newcommand{\UNIT}{\ifmmode%
\mathchoice{{\texttt{unit}}}
           {{\texttt{unit}}}
           {\mbox{\scriptsize\tt unit}}
           {\mbox{\tiny\tt unit}}
\else{\texttt{unit}}
\fi}

\newcommand{\A}{\mathcal{A}}

\renewcommand{\P}{\mathcal{P}}

\newcommand{\dom}{{\it dom}}

\newcommand{\ldb}{[\![}
\newcommand{\rdb}{]\!]}
\newcommand{\means}[1]{\ldb {#1}\rdb}
\newcommand{\union}{\cup}
\newcommand{\intersect}{\cap}
\newcommand{\proves}{\vdash}
\newcommand{\ext}[3]{[#1\mid#2\!\mapsto\!#3]}
\newcommand{\recext}[5]{[#1\mid#2\!\mapsto\!#3,#4\!\mapsto\!#5]}

\newcommand{\lam}[2]{\lambda #1.\; #2} % semantic lambda

\newcommand{\Empty}{\varnothing}

\newcommand{\tuple}[1]{\langle #1 \rangle}
\def\lpields#1{\stackrel{#1}{\longrightarrow}}
\newcommand{\annoto}[1]{\lpields {\scriptscriptstyle #1}}

\newcommand{\recdecl}[3]{\REC\;#1(#2) = #3}

\newcommand{\fix}{\textit{fix}}
\newcommand{\lub}{\sqcup}
\newcommand{\Lub}{\bigsqcup}

\newcommand{\nats}{{\mathbb{N}}}
\newcommand{\PRINC}{{\mathsf{Principals}}}
\newcommand{\PRIV}{{\mathsf{Privileges}}}
\newcommand{\ENABLE}{\mbox{\texttt{dopriv}}}
\newcommand{\CHK}{\mbox{\texttt{check}}}
\newcommand{\TEST}{\mbox{\texttt{test}}}
\newcommand{\SIGNS}{\mbox{\texttt{signs}}}

\newcommand{\safe}{\textsf{safe}}
\newcommand{\pure}{\textsf{pure}}
\renewcommand{\check}{\textsf{chk}}
\newcommand{\privs}{\textsf{privs}}
\newcommand{\Sim}{\textsf{sim}}
\newcommand{\Rel}{\textsf{rel}}

\newcommand{\stacks}{\textsf{Stacks}}

\newcommand{\letdecl}[2]{\LET#1\;\IN\;#2}
\newcommand{\LAM}[2]{{\mathtt{fun}\; #1.\; #2}}

\newcommand{\D}{D} %typing environment
\newcommand{\ty}{t} %type
\renewcommand{\th}{\theta} % annotated type

\newcommand{\n}{n} %principal
\newcommand{\p}{p} %privilege 
\newcommand{\Ps}{\Pi} %privilege set in analysis
\newcommand{\h}{h} %environment value
\newcommand{\meanss}[1]{(\![ {#1}]\!)} %stack semantics

\renewcommand{\implies}{\mathbin{\:\Rightarrow\:}}
\renewcommand{\iff}{\mathbin{\:\Leftrightarrow\:}}

{\theorembodyfont{\rmfamily} 
\newtheorem{theorem}{Theorem}[section]
\newtheorem{lemma}[theorem]{Lemma}

\newtheorem{factx}[theorem]{Fact}

\newtheorem{definition}[theorem]{Definition}
}

\newenvironment{proof}{\begin{trivlist}\item[\hskip\labelsep{\bf
Proof:}]}{\blackslug\end{trivlist}}
\newenvironment{xproof}{\begin{trivlist}\item[\hskip\labelsep{\bf  
Proof:}]}{\end{trivlist}}

\DeclareSymbolFont{italics}{OT1}{cmr}{m}{it}

\DeclareMathSymbol{a}{\mathalpha}{italics}{"61}
\DeclareMathSymbol{b}{\mathalpha}{italics}{"62}
\DeclareMathSymbol{c}{\mathalpha}{italics}{"63}
\DeclareMathSymbol{d}{\mathalpha}{italics}{"64}
\DeclareMathSymbol{e}{\mathalpha}{italics}{"65}
\DeclareMathSymbol{f}{\mathalpha}{italics}{"66}
\DeclareMathSymbol{g}{\mathalpha}{italics}{"67}
\DeclareMathSymbol{h}{\mathalpha}{italics}{"68}
\DeclareMathSymbol{i}{\mathalpha}{italics}{"69}
\DeclareMathSymbol{j}{\mathalpha}{italics}{"6A}
\DeclareMathSymbol{k}{\mathalpha}{italics}{"6B}
\DeclareMathSymbol{l}{\mathalpha}{italics}{"6C}
\DeclareMathSymbol{m}{\mathalpha}{italics}{"6D}
\DeclareMathSymbol{n}{\mathalpha}{italics}{"6E}
\DeclareMathSymbol{o}{\mathalpha}{italics}{"6F}
\DeclareMathSymbol{p}{\mathalpha}{italics}{"70}
\DeclareMathSymbol{q}{\mathalpha}{italics}{"71}
\DeclareMathSymbol{r}{\mathalpha}{italics}{"72}
\DeclareMathSymbol{s}{\mathalpha}{italics}{"73}
\DeclareMathSymbol{t}{\mathalpha}{italics}{"74}
\DeclareMathSymbol{u}{\mathalpha}{italics}{"75}
\DeclareMathSymbol{v}{\mathalpha}{italics}{"76}
\DeclareMathSymbol{w}{\mathalpha}{italics}{"77}
\DeclareMathSymbol{x}{\mathalpha}{italics}{"78}
\DeclareMathSymbol{y}{\mathalpha}{italics}{"79}
\DeclareMathSymbol{z}{\mathalpha}{italics}{"7A}

\DeclareMathSymbol{A}{\mathalpha}{italics}{"41}
\DeclareMathSymbol{B}{\mathalpha}{italics}{"42}
\DeclareMathSymbol{C}{\mathalpha}{italics}{"43}
\DeclareMathSymbol{D}{\mathalpha}{italics}{"44}
\DeclareMathSymbol{E}{\mathalpha}{italics}{"45}
\DeclareMathSymbol{F}{\mathalpha}{italics}{"46}
\DeclareMathSymbol{G}{\mathalpha}{italics}{"47}
\DeclareMathSymbol{H}{\mathalpha}{italics}{"48}
\DeclareMathSymbol{I}{\mathalpha}{italics}{"49}
\DeclareMathSymbol{J}{\mathalpha}{italics}{"4A}
\DeclareMathSymbol{K}{\mathalpha}{italics}{"4B}
\DeclareMathSymbol{L}{\mathalpha}{italics}{"4C}
\DeclareMathSymbol{M}{\mathalpha}{italics}{"4D}
\DeclareMathSymbol{N}{\mathalpha}{italics}{"4E}
\DeclareMathSymbol{O}{\mathalpha}{italics}{"4F}
\DeclareMathSymbol{P}{\mathalpha}{italics}{"50}
\DeclareMathSymbol{Q}{\mathalpha}{italics}{"51}
\DeclareMathSymbol{R}{\mathalpha}{italics}{"52}
\DeclareMathSymbol{S}{\mathalpha}{italics}{"53}
\DeclareMathSymbol{T}{\mathalpha}{italics}{"54}
\DeclareMathSymbol{U}{\mathalpha}{italics}{"55}
\DeclareMathSymbol{V}{\mathalpha}{italics}{"56}
\DeclareMathSymbol{W}{\mathalpha}{italics}{"57}
\DeclareMathSymbol{X}{\mathalpha}{italics}{"58}
\DeclareMathSymbol{Y}{\mathalpha}{italics}{"59}
\DeclareMathSymbol{Z}{\mathalpha}{italics}{"5A}
\title{A Simple Semantics and Static Analysis for Stack Inspection}
\author{Anindya Banerjee 
\institute{IMDEA Software Institute\\ Madrid, Spain}\thanks{Partially 
supported by NSF grant EIA-980635, by Madrid Regional Government MINECO Project TIN2009-14599-C03-02 Desafios and EU NoE Project 256980 Nessos.
}
\email{anindya.banerjee@imdea.org}
\and
David A.~Naumann
\institute{Stevens Institute of Technology\\ Hoboken, NJ 07030, USA}\thanks{Partially supported by NSF grants INT-9813854 and CNS-1228930.}
\email{naumann@cs.stevens.edu}
}
\def\titlerunning{A Simple Semantics and Static Analysis for Stack Inspection}
\def\authorrunning{A.~Banerjee \& D.~A.~Naumann}
\begin{document}
\maketitle

\begin{abstract}
The Java virtual machine and the .NET common language runtime feature an access control mechanism specified operationally in
terms of run-time stack inspection.
We give a denotational semantics in ``eager'' form, and show that it is equivalent to
the ``lazy'' semantics using stack inspection.  
We give a static analysis of safety, \ie, the absence of security errors,
that is simpler than previous proposals.   
We identify several program transformations that can be used to remove run-time
checks.  We give complete, detailed proofs for safety of the analysis and for the
transformations, exploiting compositionality of the eager semantics.
\end{abstract}

%%%%%%%%%%%%%%%%%%%%%%%%%%%%%%%%%%%%%%%%%%%%%%%%%%%%%%%%%%%%%%%%%%%%%%%%%

\section{Introduction}
\label{sec:intro}

System security depends in part on protecting resources through
specified access control policies. For example, a policy may allow
only some users the privilege to write the password file.  A typical
implementation of the policy found \eg, in UNIX operating systems,
involves an access control list $\A$ which associates with each user
name $\n$ their set of privileges $\A(\n)$.  When a program is running
it has an associated user, normally the user who invoked the program.
To write a file, a program for user $\n$ must make a system call, and
that system code checks whether $\A(\n)$ includes the privilege of
writing to the file. In order for users to be able to change their
passwords, the system code for this task executes in a special mode
(``setuid'' in UNIX); the effective user is the owner of the code
(say, \textit{root}) rather than the originator of the call ($\n$,
which can write some files but not the password file).

The Java and .NET platforms offer a similar but more general security system \cite{Gong99,LaMacchiaEtal02}.
Instead of code being owned by a user or
by ``the system'', there can be code from a number of sources, called
\emph{principals}, which can be offered varying degrees of trust.
Moreover, instead of associating a principal with a loadable
executable file, principals can be associated with fragments such as
class declarations.  Another refinement is that
privileges must be explicitly enabled, by an operation called
\texttt{doPrivileged}.  The intent is that a program only enables its
privileges when they are needed; this ``principle of least privilege''
\cite{Gong99} may help isolate the effect of security bugs and may
facilitate static analysis.  Before executing a dangerous operation, a
check is made that the associated privilege has been enabled and is
authorized for the current principal.  This check is specified in
terms of an implementation called \textit{stack inspection}.  Each
stack frame is marked with the principal associated with the code for
that frame, and the frame also records the privileges that have been
enabled.  This is used by procedure \texttt{checkPermission} which
inspects the current stack.    

The above description of the security system is an operational one.
While the mechanism itself is easy to understand, it may over-constrain implementations,
and it is difficult to analyze.
Analysis is of interest, e.g., to determine whether a program can exhibit security exceptions when given its expected permissions.  
Implementations of procedure calls do not always push stack, e.g., owing to inlining or tail call optimization.  
To understand the security properties achieved, and to optimize performance, we need analyses that capture the security model more abstractly.  

Our contribution is threefold: 
(i) We give a denotational semantics in ``eager'' form,
and show that it is equivalent to the ``lazy'' semantics using stack inspection.
(ii) We give a static analysis of \emph{safety}, \ie, the absence of security errors.
%that is significantly simpler than previous proposals. 
(iii) We identify several program transformations including some that can be used to remove run-time
checks. 
We give detailed proofs for the analysis and for the transformations, exploiting compositionality of the eager semantics and simplicity of program equivalence in denotational semantics.

%In the interim, simpler and more static mechanisms have been used in practice for security based on code origin, such as the static permissions in CLR have been adopted in several practical  ...
%the need for simplicity in engineering has led to more static code-based access control as in \dn{C sharp and Android and the Clearance mechanism in HAILS and Breeze}. Nonetheless, enterprise Java applications still use stack inspection.

\paragraph{Related work.}

Skalka and Smith \cite{SkalkaS:ICFP00} give an operational semantics
and use it to justify a static analysis of safety specified by a type system.  
Their type system is complicated by the choice of using a constraint system
which is the basis for a type inference algorithm.  We use a similar type system, but 
prefer to separate the specification of an analysis from algorithms to perform the
analysis.  We also include recursion in the language.
Their semantics is easily seen to model the operational descriptions of stack
inspection, but it has the usual shortcomings of operational semantics; for example, proofs ultimately go by induction on computations. 
%% ; a detailed proof is not given for
%% their safety result.

Wallach, Appel and Felten~\cite{WallachAF00} model the mechanism with
an operational semantics that manipulates formulas in a logic
of authentication \cite{AbadiBLP93}.  They show that the particular logical
deductions corresponding to \texttt{checkPermission} can be decided
efficiently, and propose an implementation called ``security passing
style'' in which the security state is calculated in advance.  The
only result proven is equivalence of the two implementations.  They do
not include recursion or higher order functions, and
the formal semantics is not made explicit.  Although the use of logic
sheds some light on the security properties achieved by the mechanism,
the approach requires a considerable amount of theory that is not
directly germane to analyzing safety or justifying optimizations.

Security passing style can be seen as a presentation of the eager
means of evaluating security checks mentioned by Li Gong~\cite{Gong99}.  
%We give a simple denotational formulation of the eager
%semantics, using only the notions of direct interest: principals and
%privileges.  For static analysis of safety, in the manner of Skalka
%and Smith, we formulate a simple syntax-directed type system and
%% prove its soundness using the denotational semantics.  We also show
%% that the denotational semantics is equivalent to the ``lazy''
%% stack-inspection semantics, and we use the semantics to justify some
%% program transformations that can be used to eliminate unnecessary
%% run-time checks.  
The eager semantics facilitates proofs, but JVM and CLR implementations
use lazy semantics which appears to have better performance \cite{Gong99,WallachAF00,LaMacchiaEtal02}.  

Pottier \emph{et al.} \cite{PottierSS05} formalize the eager semantics by a
translation into a lambda calculus, $\lambda_{\mathit{sec}}$, with operations on sets.  Using an operational semantics for the
calculus, a proof is sketched of equivalence with stack semantics.  Using a very
general framework for typing, a static analysis is given and a safety
result is sketched.  The language %% does not include recursion but
extends the language of 
Skalka and Smith \cite{SkalkaS:ICFP00} by adding permission tests.
The works \cite{SkalkaS:ICFP00} and \cite{PottierSS05} aim to replace dynamic checks by static ones, 
but do not consider program transformations as such.

This paper originated as a technical report more than a decade ago \cite{tr2001-1}.
At about the same time, and independently, Fournet and Gordon~\cite{FournetG03} 
investigated an untyped variant of $\lambda_{\mathit{sec}}$. %% language of Pottier \emph{et al.} \cite{PottierSS:ESOP01}.
They develop an equational theory that can be used to prove the correctness of code optimizations in the presence of stack inspection. A prime motivation for their work was the folklore that well-known program optimizations such as inlining and tail call elimination are invalidated by stack inpection. Their technical development uses small-step operational semantics and contextual equivalence of programs.
To prove % equivalences, they develop a form of applicative bisimilarity.
%% In addition to proving
an extensive collection of program equivalences they develop a form of applicative bisimilarity. Additionally, they prove the equivalence of the lazy and eager implementations of stack inspection. This equivalence is also proved in Skalka's Ph.D. dissertation~\cite[Theorem 4.1, Chapter 4]{skalka:phd}, where lazy and eager are termed ``backward'' and ``forward'' stack inspection respectively.

%% to ``explain how stack inspection
%% affects program behavior and code optimizations''.
%% - theirs is the main work on equational properties\\
%% - motivated by belief that tail call and inlining are invalid\\
%% - they use oper sem, eager, and contextual equality; proofs using applicative bisim  \\

%% - we are able to validate the same laws (we know none that aren't, but haven't checked all), but
%% the languages are slightly different (typed vs not) and we did not formulate a technical connection between th semantics (prove mutual consistency between our semantics and theirs)\\
%% - like us, they connect an eager semantics with stack based

Where Fournet and Gordon point out how tail call elimination can be invalidated by stack inspection, Clements and Felleisen~\cite{ClementsF04}
consider program semantics at the level of an abstract machine, namely the 
CESK machine. With this semantics they are able to show tail call optimization can be validated in full generality, and with its expected space savings. %% while still validating the source level optimizations discussed by Fournet and Gordon.
This is explored further by Ager et al~\cite{Ager05afunctional}.
They inter-derive a reduction semantics for the untyped
variant of the $\lambda_{\mathit{sec}}$-calculus, an abstract machine,
and a natural semantics, both without and with tail-call optimization.
By unzipping the context in the abstract machine, they connect these
semantics to Wallach \emph{et al.}'s security passing
style, characterize stack inspection as a
computational monad, and combine this monad with the monad of
exceptions.\footnote{Thanks to Olivier Danvy for communicating this explanation.}

We treat a simply typed language similar to $\lambda_{\mathit{sec}}$, but with recursion.
In contrast to the cited works we use a denotational semantics, which is 
straightforward; 
%Thanks to its simplicity, there is no difficulty in
%treating language constructs such as recursion; 
in fact, once the meanings
of types are specified, the rest of the specification (\ie, meanings of
expressions) follows easily.\footnote{Adding state is straightforward \cite{BanerjeeNaumann02a}, but here  we follow the cited works and confine attention to applicative expressions.}
The simplicity of our model makes it
possible to give a self-contained formal semantics and succinct but complete formal
proofs.  For example, denotational equality is a congruence simply because the semantics is compositional.
We have not formally connected the semantics with an operational one.
Adequacy seems obvious. Full abstraction is not obvious, but we have proved many of the contextual equivalences of Fournet and Gordon~\cite{FournetG03} and expect the remainder to be straightforward to show.

In addition to considering program transformations, Fournet and Gordon~\cite{FournetG03} 
 address the question of what security properties can be enforced by stack inspection.  They consider a variation that tracks history in the sense of what code has influenced the result of a computation.  Pistoia et al \cite{PistoiaBanerjeeNaumann} propose a variation that tracks implicit influences as well.  The authors \cite{BanerjeeNaumann03b} propose another combination of information flow tracking with stack inspection, using a type and effect system  where security types for functions are dependent on available permissions. 
In the interim, other code-based access mechanisms have been introduced (e.g., static permissions in the Android platform) and there have been further development of static analyses for security properties based on linear temporal logic~\cite{BessonJM01,BessonBFG04,BessonLJ05,skalkaSVh08},
but there seems to be little additional work on program equivalence in the presence of stack inspection. 

%% For security, 
%% one wants carefully checked proofs; the trusted computing base should
%% be small.  Simple, but adequate, formalizations are particularly
%% crucial for the ``proof carrying code'' approach where proof checking
%% is used for efficient, accurate static analysis of mobile code at the
%% point of deployment. The compositional nature of proofs based on
%% denotational semantics is particularly useful in this regard.

\paragraph{Outline.}

The next section explains stack inspection informally, and it introduces our language.
Section~\ref{sec:ds} gives the eager denotational semantics.
Section~\ref{sec:sa} gives the static analysis for safety, including examples and correctness proof.
Section~\ref{sec:tran} proves a number of representative program transformations.
Section~\ref{sec:use} shows how all checks can be removed from safe programs.
Section~\ref{sec:stk} gives the stack-based denotational semantics and shows that it is equivalent to the eager semantics.
%Section~\ref{sec:eg} gives examples. 
Section~\ref{sec:disc} concludes.

\section{Overview and language} 
\label{sec:ov}

Each declared procedure is associated with a principal $\n$.  We call
$\n$ the \emph{signer}, and write $\signs{\n}{e}$ for a signed
expression, because typically $\n$ is given by a cryptographic
signature on a downloaded class file.  During execution, each stack
frame is labeled with the principal that signs the function, as well
as the set $P$ of privileges that have been explicitly enabled during
execution of the function.  For our purposes, a frame is a pair
$\tuple{\n,P}$, and a nonempty stack is a list $\tuple{\n,P}::S$ with
$\tuple{\n,P}$ the top.  There should be an initial stack $S_0 =
\tuple{\n_0,\Empty}::nil$ for some designated $\n_0$.  An expression
is evaluated in a stack $S$ and with an environment $\h$ that provides
values for its free variables.

Java provides operations to enable and disable a privilege, \ie, to
add it to the stack frame or remove it.  Normally these are used in
bracketed fashion, as provided by procedure \texttt{doPrivileged} which is given a
privilege $\p$ and an expression $e$ to evaluate.  It enables $\p$,
evaluates $e$, and then disables $\p$.  Our construct is written
$\enable{\p}{e}$.  The effect of $\enable{\p}{e}$ in stack
$\tuple{\n,P}::S$ is to evaluate $e$ in stack
$\tuple{\n,P\union\{\p\}}::S $, that is, to assert $\p$ in the
current frame.
This is done regardless of whether $\p$ is authorized for $\n$.
% DN This seems to be addressed later:
% , although an
% equivalent effect can be obtained by asserting only authorized privileges.)

Java's \texttt{checkPermission} operation checks whether a certain
privilege has been enabled and is authorized for the current
principal.  Checking is done by inspecting the current stack.  Each
dangerous code fragment should be guarded by a check for an associated
privilege, so that the code cannot be executed unless the check has succeeded.  
(This can be assured by inspection of the code, or by other forms of analysis 
\cite{CentonzeFP07} but is beyond the scope of our paper.)
In our syntax, a guarded expression is written
$\chk{\p}{e}$.  The execution of an expression checked for privilege
$\p$ is to raise a security error, which we denote by $\star$, unless
the following predicate is true of $\p$ and the current stack.
\begin{eqnarray*}
\check(\p,nil) & \iff & \mathsf{false} \\
\check(\p,(\tuple{\n,P}::S)) & \iff & 
 \p\in \A(\n) \land (\p\in P \lor \check(p,S)) 
\end{eqnarray*}
That is, a privilege is enabled for a stack provided
there is some frame $\tuple{\n,P}$ with $\p\in P$ and $\p$ is authorized for
$\n$ and is authorized for all principals in frames below this one. 

% Although system code can call user code (\eg in callback patterns typical of
% object-oriented programs), malicious user code cannot exploit a privilege enabled by
% system code if the user code is not authorized for the privilege.  On the other hand,
% consider the case of system code called by user code.  One possibility is for the
% system code to enable a privilege it needs (the ``setuid'' example for passwords).
% Another, \eg for ordinary file writes, is that the user code must enable the
% privilege; then the system code will not write a file unless the user is authorized
% (the applet example).

A direct implementation in these terms requires inspecting some or all
of the stack frames.  The implementation is lazy in that no
checking is performed when a privilege is enabled, only when it is
needed to actually perform a guarded operation. On the other hand,
each check incurs a significant cost, and in secure code the checks
will never fail.  Static analysis can detect unnecessary checks, and
justify security-preserving transformations.

A stack $S$ determines a set $\privs\;S$ of enabled, authorized
privileges, to wit:
\[ \p\in\privs~S \iff \check(\p, S) \]
This gives rise to a simple form of eager semantics: instead of
evaluating an expression in the context of a stack $S$, we use
$\privs\;S$, along with the current principal which appears on top of
$S$.  The eager, stack-free semantics is given in Section~\ref{sec:ds}
and we will use this semantics exclusively in the static analysis and 
program transformations that follow in Sections~\ref{sec:sa} and ~\ref{sec:tran}.

A denotational semantics that uses explicit stacks will be deferred until Fig.~\ref{fig:ss} of Section~\ref{sec:stk}. As mentioned previously, we will then take up the equivalence of the two semantics.

The language constructs are strict in $\star$: if a subexpression
raises a security error, so does the entire expression.  In Java, security
errors are exceptions that can be caught.  Thus it is possible for a
program to determine whether a \texttt{checkPermission} operation will
succeed.  Rather than model the full exception mechanism, we
include a construct $\test{\p}{e_1}{e_2}$ which evaluates $e_1$ if
$\check(\p,S)$ succeeds in the current stack $S$, and evaluates $e_2$
otherwise.  Note that security error $\star$ is raised only by the
$\CHK$ construct, not by $\TEST$ or $\ENABLE$.

In Java, the call of a procedure of a class signed by, or otherwise
associated with, $\n$, results in a new stack frame for the method,
marked as owned by $\n$.  We model methods as function abstractions,
but whereas Skalka and Smith use signed abstractions, we include a
separate construct
$\signs{\n}{e}$.\footnote{\label{fn:pfg}%
Fournet and Gordon~\cite{FournetG03} also use a freely applicable construct for signing.
Moreover they identify principals with sets of permissions: their ``framed'' expression $R[e]$ 
is like our $\signs{\n}{e}$ for $\n$ with $\A(n) = R$.}  
Evaluation of $\signs{\n}{e}$ in stack $S$ goes by evaluating $e$ in
stack $(\tuple{\n,\Empty}::S)$.  For example, given a stack $S$, the
evaluation of the application
\[
(\LAM{x}{\signs{user}{writepass(x)}})\mbox{``myName''}
%$(\LAM{\mathit{filename}}{\signs{\mathit{sys}}{(\mathtt{open}\;\mathit{filename})}})\mbox{``myfile''}$
\]
amounts to evaluating $writepass(\mbox{``myName''})$
in the stack $(\tuple{user,\Empty}::S)$. 

We separate $\signs$ from abstractions because it helps disentangle
definitions and proofs, \eg, these constructs are treated
independently in our safety result.  On the other hand, unsigned
abstractions do not model the Java mechanism.  In our consistency
result, Theorem~\ref{thm:cons}, we show that our semantics is
equivalent to stack inspection for all \emph{standard expressions},
\ie, those in which the body of every abstraction is signed.

\subsection{Syntax and typing}
\label{sec:la}

Given are sets $\PRINC$ and $\PRIV$, and a fixed access control list
$\A$ that maps $\PRINC$ to sets of privileges.  In the grammar for
data types and expressions, $\n$ ranges over $\PRINC$ and $\p$ over
$\PRIV$.  Application associates to the left.  We include recursive
definitions for expressiveness, and simple abstractions $\LAM{x}{e}$
which, while expressible using $\mathtt{letrec}$, are easier to
understand in definitions and proofs.  For simplicity, the only
primitive type is $\BOOL$, and the only type constructor is for
functions.  Products, sums, and other primitive types can be added
without difficulty. Throughout the paper we use $\TRUE$ to exemplify the treatment of constants in general.
\begin{eqnarray*}
\ty &::=& \BOOL \mid (\ty\to\ty)\\
e &::=& \TRUE \mid x \mid \ifthenelse{e}{e_1}{e_2} \mid \\
& &     \LAM{x}{e} \mid e_1\;e_2
          \mid \letdecl{\recdecl{f}{x}{e_1}}{e_2} \mid \\
& &       \signs{\n}{e} \mid \enable{\p}{e} \mid \chk{\p}{e} \mid 
          \test{\p}{e_1}{e_2} 
\end{eqnarray*}

A signed abstraction $^{\n}\lambda x.e$ in the language of Skalka and
Smith is written $\LAM{x}{\signs{\n}{e}}$ in ours.  Our
safety result can be proved without restriction to expressions of this
form.  But for the eager semantics to be equivalent to stack
semantics, it is crucial that function bodies be signed so the
semantics correctly tracks principals on behalf of which the body of
an abstraction is evaluated.

\begin{definition}[Standard expression]%~\\
\label{def:std}
An expression is standard if for every subexpression
$\LAM{x}{e}$ or $\letdecl{\recdecl{f}{x}{e}}{e_1}$ we have that $e$ is
$\signs{\n}{e'}$ for some $\n,e'$.
\end{definition}

% For example, standard expressions disallow applications of the form 
% \begin{center}
% $(\signs{sys}{\LAM{\mathit{filename}}{\mathtt{open}\;\mathit{filename}}})\mbox{``myfile''}$ 
% \end{center}
% since, given stack $S$ and environment $\h$, the 
% meaning of such an application is the meaning of
% $(\mathtt{open}\;\mbox{``myfile''})$ in the {\it same stack} $S$ --- 
% a violation of the intended semantics of the Java security model.

Well-formed expressions are characterized by typing judgements
$\D\proves e:\ty$ which express that $e$ has type $\ty$ where free
identifiers are declared by $D$.  A typing context $D$ is a labeled
tuple of declarations $\{x_1:\ty_1,\ldots,x_k:\ty_k\}$. We write
$D,x:\ty$ for the extended context $\{x_1:\ty_1,\ldots,x_k:\ty_k,x:\ty
\}$, and $D.x_i$ for the type of $x_i$.  The typing rules are given in
Figure~\ref{fig:typing}.

\begin{figure*}
\hrule
\medskip
\[
\begin{array}{cc}
\D\proves \TRUE:\BOOL 
\\
\D,x:\ty\proves x:\ty &
\Rule{\D\proves e:\BOOL  \qquad
      \D\proves e_1:\ty \qquad 
      \D\proves e_2:\ty}
     {\D\proves\ifthenelse{e}{e_1}{e_2}:\ty} 
\\[2ex]
\Rule{\D,x:\ty_1\proves e:\ty_2}
     {\D\proves\LAM{x}{e}:\ty_1\to\ty_2} &
\Rule{\D\proves e_1: \ty_1\to\ty_2\qquad 
      \D\proves e_2: \ty_1}
     {\D\proves e_1\;e_2: \ty_2} 
\\[2ex]
\multicolumn{2}{c}{
\Rule{\D, f:\ty_1\to\ty_2, x:\ty_1
      \proves e_1:\ty_2
      \qquad
      \D, f:\ty_1\to\ty_2\proves e_2:\ty}
     {\D\proves\letdecl{\recdecl{f}{x}{e_1}}{e_2}: \ty} 
}
\\[2ex]
\Rule{\D\proves e:\ty}
     {\D\proves \signs{\n}{e}:\ty} &
\Rule{\D\proves e:\ty}
     {\D\proves \enable{\p}{e}:\ty} 
\\[2ex]
\Rule{\D\proves e:\ty}
     {\D\proves \chk{\p}{e}:\ty} &
\Rule{\D\proves e_1:\ty\qquad
      \D\proves e_2:\ty}
     {\D\proves \test{\p}{e_1}{e_2}:\ty}
\end{array}
\]
\medskip
\hrule
\medskip
\caption{Typing rules.}
\label{fig:typing}
\end{figure*}

\subsection{The password example}
\label{sec:pass}

As an example of the intended usage, we consider the problem of protecting the
password file, using a privilege $\p$ for changing password and $w$ for writing to
the password file.
The user is authorized to change passwords:
$\A(user)=\{ \p\}$.
Root is authorized to change passwords and to write the password file:
$\A(root)=\{ \p, w \}$.
Suppose $hwWrite$ is the operating system call which needs to be protected from
direct user access. The system provides the following code, which guards
$hwWrite$ with the privilege $w$.
\[ 
\begin{array}{lcl}
writepass &=&
\LAM{x}{\signs{root}{\chk{w}{hwWrite(x,\mbox{``/etc/password''})}}} \\
passwd &=& \LAM{x}{\signs{root}{\chk{p}{\enable{w}{writepass(x)}}}}  
\end{array}
\]
Consider the following user programs.
\[ 
\begin{array}{lcl}
bad1 &=& \signs{user}{writepass(\mbox{``mypass''})}\\
bad2 &=& \signs{user}{\enable{w}{writepass(\mbox{``mypass''})}} \\ 
use &=& \signs{user}{\enable{\p}{passwd(\mbox{``mypass''})}} 
\end{array}
\]
Here $bad1$ raises a security exception because $writepass$ checks for privilege $w$
which is not possessed by $user$.  The user can try to enable $w$, as in $bad2$, but
because $w$ is not authorized for $user$ the exception is still raised.  By contrast,
$use$ does not raise an exception: function $passwd$ checks for privilege $p$
which is possessed by $user$, and it enables the privilege $w$ needed by
$writepass$.  
Using transformations discussed in Section~\ref{sec:tran}, checks that never fail can
be eliminated. 
For example, the analysis will show that $use$ is safe, and the transformations will
reduce $use$ to 
%\comment{Check this}
\[ \signs{user}{ \signs{root}{ 
               hwWrite(\mbox{``mypass''},\mbox{``/etc/password''})}} 
\]

% \begin{eqnarray*}
% \mathit{passwd} &=& \LAM{x}
%            {\signs{\mathit{root}}
%                   {(\chk{\mathit{chPass}}
%                         {\mathtt{write}(x,\mbox{``passwords''})})}} \\
% \mathit{use} &=& \enable{\mathit{chPass}}
%                {\mathit{passwd}(\mbox{``secret''})} 
% \end{eqnarray*}
% Here $\mathit{use}$ raises a security exception unless $\mathit{chPass}$ is in
% $\A(root)$ and in $\A(\n)$ where $\n$ is the name of the caller of
% $\mathit{use}$.

%%%%%%%%%%%%%%%%%%%%%%%%%%%%%%%%%%%%%%%%%%%%%%%%%%%%%%%%%%%%%%%%%%%%%%%%%
\section{Denotational semantics}
\label{sec:ds}
%%%%%%%%%%%%%%%%%%%%%%%%%%%%%%%%%%%%%%%%%%%%%%%%%%%%%%%%%%%%%%%%%%%%%%%%%

This section gives the eager denotational semantics. 

\subsection{Meanings of types and type contexts}

A \emph{cpo} is a partially ordered set with least upper bounds of
ascending chains; it need not have a least element. 
Below we define, for each type $\ty$, a cpo $\means{\ty}$.
We assume that $\bot$ and $\star$ are two values not in 
$\{ \mathsf{true},  \mathsf{false}\}$ and not functions; this will ensure that 
$\{ \bot,\star\} \cap \means{\ty}=\Empty$ for all $\ty$. We will identify
$\bot$ with non-termination and $\star$ with security errors.
For cpo $C$, define $C_{\bot\star}=C\union\{\bot,\star\}$,  
ordered as the disjoint union of $C$ with $\{\star\}$, lifted with $\bot$.  
That is, for any $u,v\in C_{\bot\star}$, define
$u\leq v$ iff $u=\bot$, $u=v$, or $u$ and $v$ are in $C$ and $u\leq v$ in $C$.

We define $\means{\BOOL} = \{\mathsf{true}, \mathsf{false}\}$,
ordered by equality. We also take the powerset $\P(\PRIV)$ to be a cpo ordered by 
equality.  
Define \[ \means{\ty_1\to\ty_2}  \;  =  \;  \P(\PRIV)\to\means{\ty_1}\to
\means{\ty_2}_{\bot\star}\] where $\to$ associates to the right and 
denotes continuous function space, ordered pointwise.  Note that lubs are 
given pointwise. Also, $\means{\ty_1\to\ty_2}$ does not contain $\bot$ but it 
does have a least element, namely the constant function 
$\lam{P}{\lam{d}{\bot}}$.

Principals behave in a lexically scoped way.  By contrast, privileges are
dynamic and vary during execution; this is reflected in the semantics of the function
type.  

Let $\D=\{x_1:\ty_1,\ldots,x_k:\ty_k\}$ be a type context. Then
$\means{\D}$ is defined to be the set
$\{x_1:\means{\ty_1},\ldots,x_k:\means{\ty_k}\}$ of labeled tuples of
appropriate type.  If $\h$ is such a record, we write $\h.x_i$ for the
value of field $x_i$.  If $\D$ is the empty type context $\Empty$,
then the only element of $\means{\D}$ is the empty record $\{\}$.  For
$\h\in\means{D}$ and $d\in\means{\ty}$ we write $\ext{\h}{x}{d}$ for the
extended record in $\means{D,x:\ty}$.

\subsection{Meanings of expressions}\label{sec:mexpr}

An expression judgement denotes a function
\[
\means{\D\proves e:\ty} \; \in \; \PRINC\to\P(\PRIV)\to\means{\D}\to\means{\ty}_{\bot\star}
\]
Given a principal $\n$, a set $P\in\P(\PRIV)$ denoting privileges required
by $e$, and environment $\h\in\means{\D}$, the meaning of $\means{\D\proves e:\ty}\n P \h$
is either $\bot$ or $\star$ or an element of $\means{\ty}$.

In contrast with the work of Fournet and Gordon we do not restrict $P$ to be a subset of $\A(n)$,
though it can easily be done ---simply by giving the denotation this dependent type:
\begin{equation}\label{eq:dependDenot}
\means{\D\proves e:\ty} \; \in \; (\n:\PRINC)\to\P(\A(\n))\to\means{\D}\to\means{\ty}_{\bot\star}
\end{equation}
In programs of interest, signed at the top level, most expressions will in fact be applied to permission sets that satisfy the restriction. 
Later we observe that the restriction is need for validity of some transformations, but surprisingly few of them.  

%TBD perhaps metalanguage should use sf font throughout, e.g., check and let.

In the denotational semantics (Figure~\ref{fig:ds}), we use the metalanguage 
construct, $\mletml{d}{E_1}{E_2}$, with the following semantics: if the value 
of $E_1$ is either $\bot$ or $\star$ then that is the value of the entire let 
expression; otherwise, its value is the value of $E_2$ with $d$ bound to
the value of $E_1$. 
The semantics of if-then-else is $\star$-strict in the guard.
We also write $P\sqcup_{\n}\{\p\}$ for 
$\mifthenelse{\p\in\A(\n)}{P\union\{\p\}}{P}$.

%\[ 
%P\sqcup_{\n}\{\p\} = 
% \begin{cases}
%   P\union\{\p\} & \mbox{if $\p\in\A(\n)$} \\
%   P              & \mbox{otherwise}.
% \end{cases}
%\] 
The semantics is standard for the most part. We will only explain the
meanings of the expressions that directly concern security. In what follows,
we will assume, unless otherwise stated, that expression $e$ is signed by 
principal $\n$ and is computed with privilege set $P$ and in environment $\h$.

The meaning of $\signs{\n'}{e}$ is the meaning of $e$, signed by
$\n'$, computed with privilege set $P\intersect\A(\n')$, in $\h$. 
To illustrate the idea, consider Li Gong's
example~\cite[Section 3.11.2]{Gong99}. A game applet, \textit{applet},
has a method that calls \textit{FileInputStream} to open the file containing
the ten current high scores. In our semantics, this scenario entails finding
the meaning of $\signs{system}{FileInputStream}$, as invoked under some privilege set $P \subseteq\A(applet)$; and, this means we need to
find the meaning of \textit{FileInputStream} (\ie, whether read privileges
are enabled) under the privilege set $P\intersect\A(system)$. Assuming 
\textit{system} has all privileges, this reduces to checking if 
\textit{applet} has been granted permission to read. If it has not been granted
the permission, the file will not be read, even though it calls system code
to do so.

The meaning of $\enable{\p}{e}$ is the meaning of
$e$ computed with privilege set $P\union\{\p\}$ if $\p\in\A(\n)$, and is the 
meaning of $e$ computed with privilege set $P$ if $\p\not\in\A(\n)$. 
The meaning of $\chk{\p}{e}$ is a security error if $\p\not\in P$; otherwise, 
the meaning is that of $e$. Finally, the meaning of $\test{\p}{e_1}{e_2}$
is the meaning of $e_1$ or $e_2$ according as $\p\in P$ or $\p\not\in P$.

%TBD Although the meaning of a judgement depends on a principal as well
%as a privilege set, we do not include principals in the meaning of
%function types.  This introduces an inaccuracy in the semantics, for
%non-standard terms, but for standard terms there is no problem with
%the simplification.

\begin{figure*}
\hrule
\medskip
\[
\begin{array}{lcl}
\means{\D\proves\TRUE:\BOOL}\n P \h &=& \mathsf{true}
\\
\means{\D,x:\ty \proves x:\ty}\n P \h &=& \h.x 
\\
%% \means{\D\proves\ifthenelse{e}{e_1}{e_2}:\ty}\n P \h \\
%% \multicolumn{3}{l}{
%% \qquad = 
%% \vmlet{b\;=\;\means{\D\proves e:\BOOL}\n P \h}
%%        {\;\;\mifthenelse{b}
%%                     {\means{\D\proves e_1:\ty}\n P \h}
%%                     {\means{\D\proves e_2:\ty}\n P \h}} }
%% \\
\means{\D\proves\ifthenelse{e}{e_1}{e_2}:\ty}\n P \h &  =  &
\vmlet{b\;=\;\means{\D\proves e:\BOOL}\n P \h}
       {\mifthenelse{b}
                    {\means{\D\proves e_1:\ty}\n P \h}
                    {\means{\D\proves e_2:\ty}\n P \h}} 
\\
\means{\D\proves\LAM{x}{e}:\ty_1\to\ty_2}\n P \h &=&
\begin{array}[t]{l}
  \lam{P'\in\P(\PRIV)}{
       \lam{d\in\means{\ty_1}}}{} \\
        \means{\D, x:\ty_1\proves e:\ty_2}\n P' \ext{\h}{x}{d}
\end{array}
\\
\means{\D\proves e_1\;e_2:\ty_2}\n P \h &=&
  \vmlet{f\;=\;\means{\D\proves e_1:\ty_1\to\ty_2}\n P \h}
        {\mletml{d}{\means{\D\proves e_2:\ty_1}\n P \h}{f P d}}
\\
\multicolumn{3}{l}{
\means{\D\proves\letdecl{\recdecl{f}{x}{e_1}}{e_2}:\ty}\n P \h} \\
\multicolumn{3}{l}{
\qquad = 
\vmlet{G(g)\;=\;
\lam{P'}
    {\lam{d}
         {\means{\D,f:\ty_1\to\ty_2,x:\ty_1\proves e_1:\ty_2}\n P' 
          \recext{\h}{f}{g}{x}{d}}}}
{\means{\D,f:\ty_1\to\ty_2\proves e_2:\ty}\n P \ext{\h}{f}{\fix\;G}}} 
\\
\means{\D\proves\signs{\n'}{e}:\ty}\n P \h &=&
\means{\D\proves e:\ty}\n'(P\intersect\A(\n')) \h 
\\
\means{\D\proves\enable{\p}{e}:\ty}\n P \h &=&
\means{\D\proves e:\ty}\n(P\sqcup_{\n}\{\p\}) \h 
\\
\means{\D\proves\chk{\p}{e}:\ty}\n P \h &=&
\mifthenelse{\p\in P}{\means{\D\proves e:\ty}\n P \h}{\star}
\\ 
%% \means{\D\proves\test{\p}{e_1}{e_2}:\ty}\n P \h \\
%% \multicolumn{3}{l}{
%% \qquad = 
%% \mifthenelse{\p\in P}{\means{\D\proves e_1:\ty}\n P \h}
%%             {\means{\D\proves e_2:\ty}\n P \h}
%% }
%% % cut here
%% \\
\means{\D\proves\test{\p}{e_1}{e_2}:\ty}\n P \h 
& = & 
\mifthenelse{\p\in P}{\means{\D\proves e_1:\ty}\n P \h}
            {\means{\D\proves e_2:\ty}\n P \h}
\end{array}
\]
\medskip
\hrule
\medskip
\caption{Denotational semantics}
\label{fig:ds}
\end{figure*}

We leave it to the reader to check that the semantics of each construct is a
continuous function of the semantics of its constituent expressions,
so the semantics of recursion is well defined. 

%%%%%%%%%%%%%%%%%%%%%%%%%%%%%%%%%%%%%%%%%%%%%%%%%%%%%%%%%%%%%%%%%%%%%%%%%
\section{Static Analysis}
\label{sec:sa}
%%%%%%%%%%%%%%%%%%%%%%%%%%%%%%%%%%%%%%%%%%%%%%%%%%%%%%%%%%%%%%%%%%%%%%%%%
The denotational semantics in Section~\ref{sec:ds} gives a dynamic or run-time
view of safety; if a program is safe, its execution will not yield $\star$.
In this section, we specify a type system that statically 
guarantees safety; if a program is well-typed in the system then it is safe. 
One may utilize the static analysis for optimizing programs \eg, 
removing redundant checks of privileges at run-time.

The static analysis is specified by an extended form of typing judgement.
The idea is to give not only the type of an expression, but a
principal $\n$ and set $P$ of privileges for which the expression is
safe.  An arrow type $\ty_1\to\ty_2$ denotes functions dependent on a
set of privileges, and the static analysis uses annotated types to
track sets of privileges adequate for safety.  We adopt a Greek
notational style for types in the static analysis.  Letting $\Ps$
range over sets of privileges, annotated types, $\th$, are defined by
\[ \th ::= \BOOL \mid (\th_1\annoto{\Ps}\th_2) \]
For this syntax to be finitary, one could restrict $\Ps$ to finite sets, 
but we have no need for such restriction in our proofs.  
%Infinite sets of privileges might be useful in modelling, \eg, access to 
%files with specific names. 
An expression typed $\th_1\annoto{\Ps}\th_2$
signifies that its application may require at least the privileges $\Ps$ 
for safe execution.

\subsection{Type-based analysis}
The analysis is specified by the typing judgement $\Delta;\;\n\proves
e:\th,~\Pi$.  In words, expression $e$ signed by principal $\n$ and
typed in context $\Delta$, has (annotated) type $\th$ and is safe
provided at least the set $\Ps$ of privileges are
enabled. Figure~\ref{fig:sa} gives the specification.
\begin{figure*}
\hrule
\medskip
\[
\begin{array}{cc}
\Delta;\;\n\proves \TRUE:\BOOL, ~\Empty &
\\
\Delta,x:\th;\;\n \proves x:\th, ~\Empty &
\Rule{\Delta,x:\th_1;\;\n\proves e:\th_2, ~\Ps}
     {\Delta;\;\n\proves\LAM{x}{e}:\th_1\annoto{\Ps}\th_2, ~\Empty} 
\\[2ex]
\multicolumn{2}{c}{
\Rule{\Delta;\;\n\proves e_1: \th_1\annoto{\Ps}\th_2, ~\Ps_1 \qquad 
      \Delta;\;\n\proves e_2: \th'_1, ~\Ps_2 \qquad \th'_1\leq\th_1}
     {\Delta;\;\n\proves e_1\;e_2: \th_2, ~\Ps\union\Ps_1\union\Ps_2}
}
\\[2ex]
\multicolumn{2}{c}{
\Rule{\Delta;\;\n\proves e:\BOOL, ~\Ps_1 \qquad
      \Delta;\;\n\proves e_1:\th, ~\Ps_2 \qquad 
      \Delta;\;\n\proves e_2:\th, ~\Ps_3 \qquad}
     {\Delta;\;\n\proves \ifthenelse{e}{e_1}{e_2}:\th, 
      ~\Ps_1\union\Ps_2\union\Ps_3}
}
\\[2ex]
\multicolumn{2}{c}{
\Rule{\Delta, f:\th_1\annoto{\Ps}\th_2, x:\th_1;\;\n
      \proves e_1:\th_2, ~\Ps
      \qquad
      \Delta, f:\th_1\annoto{\Ps}\th_2;\;\n\proves e_2:\th, ~\Ps_1}
     {\Delta;\;\n\proves\letdecl{\recdecl{f}{x}{e_1}}{e_2}: \th, 
      ~\Ps\union\Ps_1}
}
\\[2ex]
\Rule{\Delta;\;\n\proves e:\th, ~\Ps}
     {\Delta;\;\n\proves \chk{\p}{e}:\th,~\Ps\union\{\p\}} &
\Rule{\Delta;\;\n\proves e:\th, \; (\Ps\sqcup_{\n}\{\p\}) }
     {\Delta;\;\n\proves \enable{\p}{e}:\th,~\Ps}
\\[2ex]
\Rule{\Delta;\;\n'\proves e:\th, \;\Ps \qquad \Ps\subseteq\A(\n') }
     {\Delta;\;\n\proves \signs{\n'}{e}:\th,~\Ps} &
\Rule{\Delta;\;\n\proves e_1:\th, ~\Ps_1\qquad
      \Delta;\;\n\proves e_2:\th, ~\Ps_2}
     {\Delta;\;\n\proves \test{\p}{e_1}{e_2}:\th,~\Ps_1\union\Ps_2}
\end{array}
\]
\medskip
\hrule
\medskip
\caption{Static analysis}
\label{fig:sa}
\end{figure*}

Constant $\TRUE$, identifiers, and anonymous functions
of the form $\LAM{x}{e}$ are all safe: they do not require any
privileges be enabled for safe execution. However, the body $e$ in
$\LAM{x}{e}$, may require a set of privileges $\Ps$ be enabled. This
is manifest in the type $\th_1\annoto{\Ps}\th_2$.
The latent privileges, $\Ps$, get exposed during an application, $e_1
e_2$. Say $e_1$ has type $\th_1\annoto{\Ps}\th_2$; if $\Ps_1$ may be
enabled during $e_1$'s execution, and $\Ps_2$ may be enabled during
$e_2$'s execution, then application itself may require $\Ps$ be
enabled; hence $\Ps\union\Ps_1\union\Ps_2$ may be enabled during the
execution of $e_1e_2$.  The application rule also uses subtyping, as discussed in the
sequel.  

The analysis for $\chk{\p}{e}$ requires that in
addition to privileges enabled for $e$, the privilege $\p$ be enabled
so that the check is safe. If $\Ps$ is the set of privileges that may
be enabled during the execution of $\enable{\p}{e}$, then $\p$ can be assumed 
to be enabled during the execution of $e$, provided $\p\in\A(\n)$.

Finally, for $\signs{\n'}{e}$ the only privileges that should be enabled are the ones
authorized for $\n'$.  Note that a signed expression can occur in a term with a different
owner, so it is not the case that $\Ps\subseteq\A(\n)$ for every derivable 
$\Delta;\n\proves e:\th,\; \Ps$.  

\subsection{Subtyping}

%% In contrast to the more complicated typing in 
Where Skalka and Smith~\cite{SkalkaS:ICFP00} use a constraint-based type system whose constraints subsequently must be solved,\footnote{Pottier \emph{et al.} \cite{PottierSS05} use unification of row variables,
  in a relatively complicated system.}
our analysis is syntax-directed. In some sense,
our system gives minimal types and privilege assumptions. (Pottier \emph{et al.}'s system \cite{PottierSS05} enjoys a principal types property also. Reasoning about minimal security contexts for code invocation is also considered in several papers by Besson \emph{et al.}~\cite{BessonJM01,BessonBFG04,BessonLJ05}.) We do not
formalize this notion, but informally it sheds light on the
specification of the analysis.  In the case of values, such as
variables and abstraction, the privilege set is empty.  In the case of
$\chk{\p}{e}$, the rule adds the checked privilege $\p$ to the
``minimal'' privileges of $e$, and similarly for the other security
constructs.  In the case of conditional, a union is formed from the
``minimal'' privileges of the constituent expressions, and the types
of the constituents are the same as the type of the conditional.  By
contrast, in the case of application $e_1 e_2$, the ``minimal'' types
and privileges for $e_1$ and $e_2$ need not match exactly.  So we
define a relation of subtyping with the informal meaning that
$\th'\leq \th$ provided the privileges required by $\th'$ are
contained in those required by $\th$.  This is significant only in
case $e_2$ has functional type, in which case the latent privileges of
$e_2$ should be among those of $e_1$.

Subtyping is defined as the least relation $\leq$ with
$\BOOL \leq \BOOL$ and, for arrow types, $\th_1\annoto{\Ps_1}\th'_1\leq\th_2\annoto{\Ps_2}\th'_2$
provided $\th_2\leq\th_1$, $\th'_1\leq\th'_2$, and $\Ps_1\subseteq\Ps_2$.

To relate the semantics to the static analysis, we need the ordinary type $\th^*$ obtained
by erasing annotations.
This is defined by induction on $\th$, to wit: 
$\BOOL^* = \BOOL$ and $(\th_1\annoto{\Ps}\th_2)^* = \th_{1}^*\to\th_{2}^*$.
It is easy to show that if $\th_1\leq\th_2$, then $\th_{1}^*=\th_{2}^*$.
%\begin{factx}
%Let $\th_1\leq\th_2$. Then $\th_{1}^*=\th_{2}^*$.
%\end{factx}

Due to subtyping, an expression can have more than one annotated type
and satisfy more than one judgement.  But a derivable judgement
$\Delta;\;\n\proves e:\th,~\Ps$ has only one derivation, which is
dictated by the structure of $e$.  Proofs in the sequel will go by
``induction on $e$'', meaning induction on the derivation of some
judgement $\Delta;\;\n\proves e:\th,~\Ps$.

%%%%%%%%%%%%%%%%%%%%%%%%%%%%%%%%%%%%%%%%%%%%%%%%%%%%%%%%%%%%%%%%%%%%%%%

%\subsection{The password example analyzed}
%\label{sec:peg}
\subsection{Examples}

For any $\n$, the expressions in the password example (Section~\ref{sec:pass}) can be analyzed as follows.
\[
\begin{array}{l}
\Empty; \ \n\proves writepass: string\annoto{\{w\}} void,\ \Empty \\
\Empty; \ \n\proves passwd: string\annoto{\{p\}} void,\ \Empty \\
\Empty; \ \n\proves use:  void,\ \Empty 
\end{array}
\]
This confirms that $use$ is safe.
On the other hand, there is no $\Ps$ such that
$\Empty; \ \n\proves bad1: void,\ \Ps$ or
$\Empty; \ \n\proves bad2: void,\ \Ps$.
Such $\Ps$ must satisfy $w\in\Ps$ for the application of $writepass$, owing to the rules for application and for $\enable{}{}$.  And $\Ps$ must satisfy $\A(user)\subseteq \Ps$ by the rule for signs.

%%%%%%%%%%%%%%%%%%%%%%%%%%%%%%%%%%%%%%%%%%%%%%%%%%%%%%%%%%%%%%%%%%%%%%%
%\subsection{Examples}
%\label{sec:eg}
%%%%%%%%%%%%%%%%%%%%%%%%%%%%%%%%%%%%%%%%%%%%%%%%%%%%%%%%%%%%%%%%%%%%%%%

Here is another example, inspired by ones in Skalka and Smith \cite{SkalkaS:ICFP00}.
Define the following standard expressions:
\begin{eqnarray*}
lp &=& \LAM{f}
           {\signs{\n}
                  {(\LAM{x}
                        {\signs{\n}
                               {(\enable{\p}{(f\ x)})}})}} \\
cp &=& \LAM{x}{\signs{\n}
                     {(\chk{\p}{x})}}
\end{eqnarray*}
The reader can verify that one analysis for $cp$ is given by the 
typing $\Delta;\;\n\proves cp:(\BOOL\annoto{\{\p\}}\BOOL), \Empty$
and that the typing demands $\p\in\A(\n)$. Similarly, the reader can verify 
that one possible analysis for $lp$ is given by the typing
$\Delta;\;\n\proves lp:(\BOOL\annoto{\{\p\}}\BOOL)\annoto{\Empty}
(\BOOL\annoto{\Empty}\BOOL), \Empty$.

For all $P\in\P(\PRIV)$, for all $\h:\Delta^*$, we can show (omitting types and
some steps),
\[
\begin{array}{lcl}
\means{lp}\n P\h 
&=& \lam{P_1}
        {\lam{d_1}
             {\lam{P_2}
                  {\lam{d_2}
                       {\means{\enable{\p}{f\;x}}\n(P_2\intersect\A(\n))
                                      [\h\mid f\mapsto d_1, x\mapsto d_2]}}}}\\
&=& \{\mathrm{letting}\;P_3=P_2\intersect\A(\n)\}\\
& & \lam{P_1}
        {\lam{d_1}
             {\lam{P_2}
                  {\lam{d_2}
                       {d_1(P_3\sqcup_{\n}\{\p\})d_2}}}}\\
\means{cp}\n P\h
&=&
\lam{P'_1}
    {\lam{d'_1}
         {\mifthenelse{\p\in(P'_1\intersect\A(\n))}
                      {\means{x}\n(P'_1\intersect\A(\n))\ext{\h}{x}{d'_1}}
                      {\star}}}\\
&=& \lam{P'_1}
    {\lam{d'_1}
         {\mifthenelse{\p\in(P'_1\intersect\A(\n))}{d'_1}{\star}}}
\end{array}
\]
Let $F=\means{lp}\n P\h$, let $d=\means{cp}\n P\h$, and
let $G=\means{(lp\;cp)}\n P\h$. Then
\[
\begin{array}{lcl}
\means{(lp\;cp)}\n P\h
&=&
FPd\\
&=&
\lam{P_2}
    {\lam{d_2}
         {\mifthenelse{\p\in((P_3\sqcup_{\n}\{\p\})\intersect\A(\n))}
                      {d_2}{\star}}} \\
%&=&\{\mathrm{because\;}\p\in\A(\n)\}\\
&= &
\lam{P_2}{\lam{d_2}{d_2}}
 \quad \mathrm{because\;}\p\in\A(\n)
\\
\means{(lp\;cp)\TRUE}\n P\h
&=&
GP(\means{\TRUE}\n P\h)\\
&=&
\mathsf{true}
\end{array}
\]
Hence $(lp\;cp)\TRUE$ is safe in any environment and typable as 
$\Delta;\;\n\proves (lp\;cp)\TRUE: \BOOL, \Empty$.

%%%%%%%%%%%%%%%%%%%%%%%%%%%%%%%%%%%%%%%%%%%%%%%%%%%%%%%%%%%%%%%%%%%%%%%

\subsection{Safety of the analysis}
\label{sec:pos}

\begin{theorem}[Safety]
\label{thm:safe}
Suppose $\Empty;\;\n\proves e:\th,~\Ps$ is derivable.
Then for all $P\in\P(\PRIV)$ with $\Ps\subseteq P$, it is the case that
$\means{\Empty\proves e:\th^*}\n P \{\} \neq \star$.
\end{theorem}
\begin{proof}
Immediate consequence of Lemma \ref{lem:safe} below.
\end{proof}

In order to serve as an adequate induction hypothesis, the lemma
strengthens the theorem by allowing judgements with non-empty
contexts.  But this is not enough.  Values at arrow types are
functions that depend on privilege sets. As induction hypothesis for
the case of application we require these functions be safe with
respect to the privilege set $\Ps$ annotating their type.

\begin{definition}
For each annotated type $\th$ the predicate $\safe~\th$ on 
$\means{\th^*}_{\bot\star}$ is defined as follows: \\
$\safe~ {\th}(\bot) \iff \mathsf{true}$ and 
$\safe~ {\th}(\star)\iff \mathsf{false}$
for all $\th$.  For values other than $\bot$ and $\star$, the definition is by
induction on structure of $\th$.
\[
\begin{array}{lcl}
\safe~ {\BOOL}(b) &\iff & \mathsf{true} \\
\safe~(\th_1\annoto\Ps\th_2)(f) &\iff &
\forall P\in\P(\PRIV) .\forall d\in\means{\th_1^*} .\; \\
&&  \Ps\subseteq P \land 
\safe~ {\th_1}(d)\implies \safe~ {\th_2}(f P d)
\end{array}
\]
%For annotated type environment $\Delta$, 
The predicate  
$\safe~\Delta$ on $\means{\Delta^*}$ is defined by
$ \safe~\Delta(h) \iff \forall x\in\dom(\h) . \safe(\Delta.x)(\h.x) $.
Recall that $\h.x\neq\bot$ and $\h.x\neq\star$, because $\bot\not\in\means{\ty}$ and
$\star\not\in\means{\ty}$, for all $\ty$.
\end{definition}

\begin{factx}\label{fact:safeMono}
  $\th\leq \th'$ and $\safe~\th~d$ imply $\safe~\th'~d$.
\end{factx}
\begin{proof}
By induction on derivation of $\th\leq\th'$.  The result is clear for
$\BOOL\leq\BOOL$. For $(\th_1\annoto{\Ps}\th_2)\leq(\th'_1\annoto{\Ps'}\th'_2)$,
assume $\safe~(\th_1\annoto{\Ps}\th_2)~f$. 
To show $\safe~(\th'_1\annoto{\Ps'}\th'_2)~f$, consider 
any $P\in\P(\PRIV)$, such that $\Ps'\subseteq P$, and any 
$d\in\means{{\th'_1}^*}$ with $\safe~\th'_1~d$. From the subtyping, we know
that $\Ps\subseteq\Ps'$, hence $\Ps\subseteq P$. Moreover, by induction
on derivation of $\th'_1\leq\th_1$, we obtain $\safe~\th'_1~d$ implies
$\safe~\th_1~d$. Hence from assumption $\safe~(\th_1\annoto{\Ps}\th_2)~f$, 
we obtain $\safe~\th_2(f P d)$ holds. Now by induction on derivation 
$\th_2\leq\th'_2$, we obtain $\safe~\th'_2(f P d)$.  
\end{proof}

\begin{lemma}
\label{lem:adm}
The predicate $\safe$ preserves lubs. That is, for any $\th$,
let $u:\nats\to\means{\th^*}_{\bot\star}$ be an ascending chain. Then,
$\forall i.\safe~\th~(u_i)$ implies $\safe~\th~(\Lub_{i}u_i)$.

\end{lemma}
\begin{proof}
By structural induction on $\th$. When $\th=\BOOL$, 
the assumption $\safe~\th~(u_i)$ implies $u_i\neq\star$ for each $i$, 
so $\Lub_{i}u_i$ is 
$\mathsf{true}$ or $\mathsf{false}$ or $\bot$.  Thus the result holds by definition 
$\safe$. \\
When $\th=(\th_1\annoto{\Ps}\th_2)$, assume $P\in\P(\PRIV)$ and 
$d\in\means{\th_1^*}$, such that $\Ps\subseteq P$ and 
$\safe~\th_{1}(d)$. Then, from assumption $\safe~(\th_1\annoto{\Ps}\th_2)~u_i$
we obtain $\safe~\th_2~(u_{i} P d)$ holds for every $i$. By the induction 
hypothesis on $\th_2$, we get $\safe~\th_2~(\Lub_{i}(u_{i} P d))$. 
Lubs are pointwise, so we get $\safe~\th_2~((\Lub_{i}u_{i}) P d)$.
\end{proof}
\begin{lemma}
\label{lem:safe}
Suppose $\Delta;\;\n\proves e:\th,~\Ps$ is derivable.
Then for all $P\in\P(\PRIV)$, for all $\h\in\means{\Delta^*}$, if $\safe~ {\Delta}(\h)$
and $\Ps\subseteq P$ then  
$\safe~\th~(\means{\Delta^*\proves e:\th^*}\n P \h) $.
\end{lemma}

\begin{sloppypar}
Theorem~\ref{thm:safe} follows from the lemma because 
$\safe~ {\Empty} \{\}$ and  
$\safe~ {\th} (\means{\Empty\proves e:\th^*}\n P \{\}) $ 
implies 
$\means{\Empty\proves e:\th^*}\n P \{\} \neq \star$.
\end{sloppypar}

Another consequence of the lemma is that the language admits
additional constants at all types, declared in an initial context
$D_0$, provided the corresponding initial environment assigns a safe
meaning to each identifier in $D_0$.

\begin{xproof} of Lemma.
Go by induction on the typing derivation, $\Delta;\;\n\proves e:\th,~\Ps$.
Throughout, we assume $P\in\P(\PRIV)$ and $\h\in\means{\Delta^*}$ and $\safe\
{\Delta}(\h )$ and  $\Ps\subseteq P$, and also let
$u=\means{\Delta^*\proves e:\th^*}\n P \h$ for each case of $e$.
\begin{itemize}
%\item \fbox{Case of $\TRUE$}~\\[2mm]
\item Case $\TRUE$:
Then, $u = \mathsf{true}$ so $\safe~ {\BOOL}(u)$ by definition $\safe$.
%\\
%%%%%%%%%%%%%%%%%%%%%%%%%%%%%%%%%%%%%%%%%%%%%%%%%%%%%%%%%%%%%%%%%%%%%%%%
%\item \fbox{Case of $x$}~\\[2mm]
\item Case $x$:
Then, $u=\h.x$ and $\safe~ {\th}(\h.x)$ follows, by definition $\safe$, 
from the assumption $\safe~ {\Delta}(\h )$.
%\\
%%%%%%%%%%%%%%%%%%%%%%%%%%%%%%%%%%%%%%%%%%%%%%%%%%%%%%%%%%%%%%%%%%%%%%%%
%\item \fbox{Case of $\ifthenelse{e}{e_1}{e_2}$}~\\[2mm]
\item Case $\ifthenelse{e}{e_1}{e_2}$:
Then $\Ps_1\union\Ps_2\union\Ps_3\subseteq P$, and 
\[ u=\mifthenelse{b}{\means{\Delta^*\proves e_1:\th^*}\n P \h}
                  {\means{\Delta^*\proves e_2:\th^*}\n P \h}\]
 where
$b=\means{\Delta^*\proves e:\BOOL}\n P \h$. By the induction hypothesis
on the typing derivation of $e$, noting that $\Ps_1\subseteq P$, we have $\safe\
\BOOL(b)$ and hence $b\neq\star$. 
If $b=\bot$ then $u=\bot$ and $\bot$ is safe.
Otherwise, $b=\mathsf{true}$ or $b=\mathsf{false}$. In the
former case, by the induction hypothesis on the typing derivation of 
$e_1$, noting that $\Ps_2\subseteq P$, we have 
$\safe~ {\th}(u)$. The case of $b=\mathsf{false}$ is symmetric.
%\\
%%%%%%%%%%%%%%%%%%%%%%%%%%%%%%%%%%%%%%%%%%%%%%%%%%%%%%%%%%%%%%%%%%%%%%%%%
%\item \fbox{Case of $\LAM{x}{e}$}~\\[2mm]
\item Case $\LAM{x}{e}$:
%$\Rule{\Delta,x:\th_1;\;\n\proves e:\th_2, ~\Ps}
%      {\Delta;\;\n\proves\LAM{x}{e}:\th_1\annoto{\Ps}\th_2, ~\Empty}$}
%~\\[2mm]
Then 
$u=\lam{P'}
       {\lam{d}
            {\means{\Delta^*,x:\th_{1}^*\proves e:\th_{2}^*}\n P'\ext{\h}{x}{d}}}$.
Thus $u\neq\star$. To prove $\safe~(\th_{1}\annoto{\Ps}\th_{2})(u)$, consider any 
$P''\in\P(\PRIV)$ and any $d'\in\means{\th_{1}^*}$ such that 
$\Ps\subseteq P''$ and $\safe~ {\th_{1}}(d')$, to show $\safe~\th_2(u P'' d')$.
By semantics, $u P'' d' = 
\means{\Delta^*,x:\th_{1}^*\proves e:\th_{2}^*}\n P''\ext{\h}{x}{d'}$, 
so the induction hypothesis for $e$ yields $\safe~\th_2(u P'' d')$ provided
that $\Ps\subseteq P''$ and $\safe~ {(\Delta,x:\th_{1})}\ext{\h}{x}{d'}$.
We have $\Ps\subseteq P''$ by assumption, and 
$\safe~ {(\Delta,x:\th_{1})}\ext{\h}{x}{d'}$ follows from $\safe~ {\Delta}(\h )$ and 
$\safe~ {\th_{1}}(d')$.
%\\
%%%%%%%%%%%%%%%%%%%%%%%%%%%%%%%%%%%%%%%%%%%%%%%%%%%%%%%%%%%%%%%%%%%%%%%
%\item \fbox{Case of $e_1\;e_2$}~\\[2mm]
\item Case $e_1\;e_2$:
%$\Rule{\Delta;\;\n\proves e_1: \th_1\annoto{\Ps}\th_2, ~\Ps_1 \qquad 
%      \Delta;\;\n\proves e_2: \th'_1, ~\Ps_2 \qquad \th'_1\leq\th_1}
%      {\Delta;\;\n\proves e_1\;e_2: \th_2, ~\Ps\union\Ps_1\union\Ps_2}$}
%~\\[2mm]
Let $f=\means{\Delta^*\proves e_1:\th_{1}^*\to\th_{2}^*}\n P \h$ and 
$d=\means{\Delta^*\proves e_2:{\th'_1}^*}\n P \h$, so that $u=f P d$.  
(Recall that $\th'_1\leq\th_1$ implies ${\th'_1}^*=\th_{1}^*$ so the
application $f P d$ makes sense.)  From safety of $\h$ and the assumption
$\Ps\union\Ps_1\union\Ps_2\subseteq P$,  
we get by induction on $e_1$ that $\safe~(\th_{1}\annoto{\Ps}\th_{2})(f)$, 
and we get $\safe~ {\th'_1}(d)$ by induction on $e_2$. 
By $\th'_1\leq \th_1$ and Fact~\ref{fact:safeMono} we have $\safe~
{\th'_1}(d)\implies \safe~ {\th_1}(d)$. 
Then by definition $\safe~(\th_{1}\annoto{\Ps}\th_{2})(f)$ we get
$\safe~ {\th_2}(f P d)$.
%\\
%%%%%%%%%%%%%%%%%%%%%%%%%%%%%%%%%%%%%%%%%%%%%%%%%%%%%%%%%%%%%%%%%%%%%%%%%
%\item \fbox{Case of $\letdecl{\recdecl{f}{x}{e_1}}{e_2}$}~\\[2mm]
\item Case $\letdecl{\recdecl{f}{x}{e_1}}{e_2}$:
%$\Rule{\Delta, f:\th_1\annoto{\Ps}\th_2, x:\th_1;\;\n
%      \proves e_1:\th_2, ~\Ps
%      \qquad
%      \Delta, f:\th_1\annoto{\Ps}\th_2;\;\n\proves e_2:\th, ~\Ps_1}
%     {\Delta;\;\n\proves\letdecl{\recdecl{f}{x}{e_1}}{e_2}: \th, 
%      ~\Ps\union\Ps_1}$}~\\[2mm]
Then, $\Ps\union\Ps_1\subseteq P$.\\
Now $u=\means{\Delta^*,f:\th_1^*\to\th_2^*\proves e_2:\th^*}\n P 
\ext{\h}{f}{\fix\;G}$, where
\[ 
G(g)=\lam{P'}
         {\lam{d}
              {\means{\Delta^*,f:\th_1^*\to\th_2^*,x:\th_1^*
               \proves e_1:\th_2^*}\n P'\recext{\h}{f}{g}{x}{d}}}
\]
To get $\safe~ {\th}(u)$ by induction for $e_2$, we need 
$\Ps_1\subseteq P$ and 
\[
\safe~ {(\Delta,f:\th_1\annoto{\Ps}\th_2)}\ext{\h}{f}{\fix\;G} \] 
The former follows from the assumption $\Ps\union\Ps_1\subseteq P$.
The latter follows from assumption, $\safe~ {\Delta}(\h )$, and 
$\safe~(\th_1\annoto{\Ps}\th_2)(\fix\;G)$.
We proceed to show safety of $\fix\;G$.  

Now $\fix\;G=\Lub_{i}g_i$, where $g_0=\lam{P''}{\lam{d\in\means{\th_1^*}}{\bot}}$ and
$g_{i+1}=G(g_i)$. 
And, $\safe~(\th_1\annoto{\Ps}\th_2)(\fix\;G)$ is a consequence of the 
following claim:
\begin{eqnarray}
  \label{eq:one}
& \forall i.\; \safe~(\th_1\annoto{\Ps}\th_2)(g_i) 
\end{eqnarray}
Then from Lemma~\ref{lem:adm}, we get 
$\safe~(\th_1\annoto{\Ps}\th_2)(\Lub_{i}g_i)$. 
It remains to show (\ref{eq:one}), for which we proceed by induction on $i$.\\
Base case: Show $\safe~(\th_1\annoto{\Ps}\th_2)(g_0)$. Assume any 
$P''\in\P(\PRIV)$ and any $v\in\means{\th_1^*}$, such that 
$\Ps\subseteq P''$ and $\safe~ {\th_1}(v)$. Then $g_{0}P''v=\bot\neq\star$ 
and $\safe~ {\th_2}(g_{0}P''v)$ holds. \\
Induction step: Assume $\safe~(\th_1\annoto{\Ps}\th_2)(g_i)$, to show
$\safe~(\th_1\annoto{\Ps}\th_2)(g_{i+1})$. \\
Now $g_{i+1}=G(g_i)=
\lam{P'}{\lam{d}{\means{\Delta^*,f:\th_1^*\to\th_2^*,x:\th_1^*
\proves e_1:\th_2^*}\n P \recext{\h}{f}{g_i}{x}{d}}}$. Assume any
$P''\in\P(\PRIV)$ and $v\in\means{\th_1^*}$, such that $\Ps\subseteq P''$
and $\safe~ {\th_1}(v)$.  
Then 
\[
g_{i+1}P''(v)=\means{\Delta^*,f:\th_1^*\to\th_2^*,x:\th_1^*
\proves e_1:\th_2^*}\n P'' \recext{\h}{f}{g_i}{x}{v}
\]
Note that
$\safe~(\Delta,f:\th_1\annoto{\Ps}\th_2,x:\th_1)\recext{\h}{f}{g_i}{x}{v}$.
Therefore, by the main induction hypothesis on the typing derivation 
$\Delta, f:\th_1\annoto{\Ps}\th_2, x:\th_1;\;\n
      \proves e_1:\th_2, ~\Ps$, since $\Ps\subseteq P$, we obtain
$\safe~ {\th_2}(g_{i+1}P''v)$.
%
%To prove (\ref{eq:two}),
%consider any $P'\in\P(\PRIV)$ and $v\in\means{\th_1}$, such that
%$\Ps\subseteq P'$ and $\safe~ {\th_1}(v)$. Because lubs are pointwise, 
%$(\Lub_{i\geq 0}g_i)P'v = \Lub_{i\geq 0}(g_{i}P'v))$. Assume for each $i$,
%$\safe~(\th_1\annoto{\Ps}\th_2)(g_i)$. 
%Then $g_{i}P'v\neq\star$, for every $i$. Moreover, for each $i$, 
%$g_{i}P'v\in\means{\th_2^*}_{\bot\star}$; thus,
%$\Lub_{i\geq 0}(g_{i}P'v)\neq\star$, as $\star$ is not a join of two 
%non-bottom elements.
%
%TBD: we need safety of the join, not just $\neq\star$.  
%Perhaps the admissibility of $\safe$ should be made a separate lemma.
%\\
%%%%%%%%%%%%%%%%%%%%%%%%%%%%%%%%%%%%%%%%%%%%%%%%%%%%%%%%%%%%%%%%%%%%%%
%\item \fbox{Case of $\signs{\n'}{e}$}~\\[2mm]
\item Case $\signs{\n'}{e}$:
%$\Rule{\Delta;\;\n'\proves e:\th, ~(\Ps\intersect\A(\n'))}
%      {\Delta;\;\n\proves \signs{\n'}{e}:\th,~\Ps}$}~\\[2mm]
Then $\Ps\subseteq P$ and 
$u=\means{\Delta^*\proves e:\th^*}\n'(P\intersect\A(\n')) \h$. 
The induction hypothesis on the typing derivation
of $e$ can be used to obtain $\safe~ {\th}(u)$, 
because $\Ps\subseteq(P\intersect\A(\n'))$ which follows from 
assumption $\Ps\subseteq P$ and side condition $\Ps\subseteq \A(\n')$
on the antecedent $\Delta;\;\n'\proves e:\th, ~\Ps$ of 
$\Delta;\;\n\proves \signs{\n'}{e}:\th,~\Ps$.
%\\
%%%%%%%%%%%%%%%%%%%%%%%%%%%%%%%%%%%%%%%%%%%%%%%%%%%%%%%%%%%%%%%%%%%%%%
%\item \fbox{Case of $\enable{\p}{e}$}~\\[2mm]
\item Case $\enable{\p}{e}$:
%$\Rule{\Delta;\;\n\proves e:\th, ~\Ps\sqcup_{\n}\{\p\}}
%      {\Delta;\;\n\proves \enable{\p}{e}:\th,~\Ps}$}~\\[2mm]
Then $\Ps\subseteq P$ and 
$u=\means{\Delta^*\proves e:\th^*}\n(P\sqcup_{\n}\{\p\}) \h$.
By the induction hypothesis for $e$, noting that 
$(\Ps\sqcup_{\n}\{\p\})\subseteq(P\sqcup_{\n}\{\p\})$, we have $\safe~ {\th}(u)$.
%\\
%%%%%%%%%%%%%%%%%%%%%%%%%%%%%%%%%%%%%%%%%%%%%%%%%%%%%%%%%%%%%%%%%%%%%%%
%\item \fbox{Case of $\chk{p}{e}$}~\\[2mm]
\item Case $\chk{p}{e}$: 
%$\Rule{\Delta;\;\n\proves e:\th, ~\Ps}
%      {\Delta;\;\n\proves \chk{\p}{e}:\th,~\Ps }$}~\\[2mm]
\begin{sloppypar}
Then $\Ps\union\{\p\}\subseteq P$, hence $\p\in P$.
Now \( u=\mifthenelse{\p\in P}{\means{\Delta^*\proves e:\th^*}\n P \h}{\star} \).
Since $\p\in P$, we have,
$u=\means{\Delta^*\proves e:\th^*}\n P \h$ and, by the induction
hypothesis on the typing derivation of $e$, we have $\safe~ {\th}(u)$.
\end{sloppypar}
%\\
%%%%%%%%%%%%%%%%%%%%%%%%%%%%%%%%%%%%%%%%%%%%%%%%%%%%%%%%%%%%%%%%%%%%%%%
%\item \fbox{Case of $\test{\p}{e_1}{e_2}$}~\\[2mm]
\item Case $\test{\p}{e_1}{e_2}$:
%$\Rule{\Delta;\;\n\proves e_1:\th, ~\Ps_1\qquad
%      \Delta;\;\n\proves e_2:\th, ~\Ps_2}
%      {\Delta;\;\n\proves \test{\p}{e_1}{e_2}:\th,~\Ps_1\union\Ps_2}$}~\\[2mm]
Then $\Ps_1\union\Ps_2\subseteq P$ and 
\[ u=\mifthenelse{\p\in P}{\means{\Delta^*\proves e_1:\th^*}\n P \h}
               {\means{\Delta^*\proves e_2:\th^*}\n P \h} \]
We have two cases. Suppose $\p\in P$. Then,
by induction hypothesis on typing derivation of $e_1$ and noting that
$\Ps_1\subseteq P$, we have $u=\means{\Delta^*\proves e_1:\th^*}\n P \h$ and $\safe~
{\th}(u)$. The case where $\p\not\in\P$, is symmetric. 
\blackslug
%%%%%%%%%%%%%%%%%%%%%%%%%%%%%%%%%%%%%%%%%%%%%%%%%%%%%%%%%%%%%%%%%%%%%%%%%%%%
\end{itemize}
\end{xproof}
%%%%%%%%%%%%%%%%%%%%%%%%%%%%%%%%%%%%%%%%%%%%%%%%%%%%%%%%%%%%%%%%%%%%%%%

%%%%%%%%%%%%%%%%%%%%%%%%%%%%%%%%%%%%%%%%%%%%%%%%%%%%%%%%%%%%%%%%%%%%%%%%%
\section{Some program transformations}
%%%%%%%%%%%%%%%%%%%%%%%%%%%%%%%%%%%%%%%%%%%%%%%%%%%%%%%%%%%%%%%%%%%%%%%%%
\label{sec:tran}
Using the eager semantics it is straightforward to justify program transformations
that can be used for optimization.  
This section shows how checks can be eliminated from the password example 
and also considers the proofs of some primitive equations from Fournet and Gordon's work~\cite{FournetG03}.

%%%%%%%%%%%%%%%%%%%%%%%%%%%%%%%%%%%%%%%%%%%%%%%%%%%%%%%%%%%%%%%%%%%%%%%%%
\subsection{Transformations that eliminate checks}\label{sec:elimcheck}
%%%%%%%%%%%%%%%%%%%%%%%%%%%%%%%%%%%%%%%%%%%%%%%%%%%%%%%%%%%%%%%%%%%%%%%%%

First, we list a series of program transformations that move checking of 
privileges ``outwards''.
\[
\begin{array}{lcl}
\ifthenelse{e}{\chk{\p}{e_1}}{\chk{\p}{e_2}}
&=&
\chk{\p}{\ifthenelse{e}{e_1}{e_2}}\\
e_1(\chk{\p}{e_2}) 
&=&
\chk{\p}{e_1e_2}\\
\test{\p}{e_1}{\chk{\p}{e_2}}
&=&
\chk{\p}{\test{\p}{e_1}{e_2}}\\
\test{\p'}{\chk{\p}{e_1}}{\chk{\p}{e_2}}
&=&
\chk{\p}{\test{\p'}{e_1}{e_2}}\\
\letdecl{\recdecl{f}{x}{e_1}}{\chk{\p}{e_2}}
&=&
\chk{\p}{\letdecl{\recdecl{f}{x}{e_1}}{e_2}}\\
\chk{\p}{\chk{\p}{e}}
&=&
\chk{\p}{e} 
% \\ 
% \signs{n}{\signs{n}{e}} & = & \signs{n}{e} 
\end{array}
\]
These are unconditional equalities, as the reader can verify using the denotational
semantics (Figure~\ref{fig:ds}).
We emphasize that this means extensional equality of the functions denoted by the two sides.  
In particular, $e=e'$ means 
$\means{D\proves e:\ty}\n P \h =
 \means{D\proves e':\ty}\n P \h$ for all $\n,P,\h$.
Corresponding to Fournet and Gordon one may consider a slightly weaker notion of equality where $P$ ranges over subsets of $\A(\n)$, as indicated by (\ref{eq:dependDenot}) in Section~\ref{sec:mexpr}.
Later we encounter one transformation that only holds for that weaker equality.

Once checks have been moved outward, some can be eliminated.
% In some cases we can remove redundant checks using Lemma~\ref{lem:chk}.
% \comment{Dave doesn't see how.
% Perhaps what we want is a result like Lemma~\ref{lem:pur} licensed by static
% analysis; then instead of purity we use static analysis for the theorem?
% }
% %Let $\p_1,\ldots,\p_n$ be an enumeration of $\Ps$, and let $\chk{\Ps}{e}$ abbreviate 
% %\[ \chk{\p_1}{\chk{\p_2}{\ldots\FOR\;\chk{\p_n}{e}}} \] 
% \begin{lemma}
% \label{lem:chk}
% Suppose $\Empty;\;\n\proves \chk{\p}{e}:\th,~\{\p\}$ is derivable.
% Then for all $P\in\P(\PRIV)$ with $\p\in P$, it is the case that
% \\
% $\means{\Empty\proves \chk{\p}{e}:\th^*}\n P \{\} =
%  \means{\Empty\proves e:\th^*}\n P \{\}$.
% \end{lemma}
% \begin{proof}
% By Theorem~\ref{thm:safe} and semantics of $\CHK$.
% \end{proof}
To eliminate a check, it must be known definitely to succeed, \eg, because it has been
enabled for an authorized principal.  We give an example transformation of this
kind in Theorem~\ref{thm:elim}, formulated in terms of the following notions
concerning expressions that do not depend on privilege $\p$.

\begin{definition}[p-purity]
\begin{sloppypar}
An expression $e$ is $\p$-pure if $e$ has no sub-expressions of the form
$\chk{\p}{e'}$ or $\test{\p}{e'}{e''}$.  
\end{sloppypar}

For each type $\ty$ we define semantic $\p$-purity as a predicate 
$\pure~\p~{\ty}$ on $\means{\ty}_{\bot\star}$,  as follows: \\
$\pure~\p~{\ty}(\bot) \iff \mathsf{true}$ and 
$\pure~\p~{\ty}(\star) \iff \mathsf{true}$
for all $\ty$.  For values other than $\bot$ and $\star$, the definition is by
induction on structure of $\ty$.
\[
\begin{array}{lcl}
\pure~\p~ {\BOOL}(b) &\iff & \mathsf{true} \\
\pure~\p~(\ty_1\to \ty_2)(f) &\iff &
\forall P\in\P(\PRIV) .\forall d\in\means{\ty_1} .\; \\
&& \pure~\p~ {\ty_1}(d)\implies \pure~\p~ {\ty_2}(f P d)
\land f P d = f (P-\{\p\}) d 
\end{array}
\]
Finally, for environment $h\in\means{D}$ we define
$\pure~\p~D(\h)$ iff $\pure~\p~\ty(\h.x)$ for all $x:\ty$ in $D$.
\end{definition}

\begin{lemma}
\label{lem:purlub}
Suppose  $u:\nats\to \means{\ty_1\to\ty_2}$ is an ascending chain.
Then $\forall i. \pure~\p~(\ty_1\to \ty_2)(u_i)$ implies 
$\pure~\p~(\ty_1\to \ty_2)(\lub_i u_i)$.
\end{lemma}
\begin{proof}
By definition of pure and since joins are given pointwise.  
\end{proof}

\begin{lemma}
\label{lem:pur}
  If $e$ is $\p$-pure and typable as $\D\proves e:\ty$, then for all $\n,P,\h$
with $\pure~\p~D(\h)$ we have
\[ \means{ D\proves e:\ty} \n P \h =
   \means{ D\proves e:\ty} \n (P-\{\p\}) \h 
\]
and $\pure~\p~\ty(\means{ D\proves e:\ty} \n P \h)$.
\end{lemma}
\begin{xproof}
By induction on $e$.  We observe for any $\D,\n,P,\h$ with 
$\pure~\p~\D~(\h)$
\begin{itemize}
\item Case $\TRUE$: The equation is direct from the
 semantics, which is independent of $P$.  For $\p$-purity of $\mathsf{true}$,
the result holds by definition of $\pure~\p~\BOOL$.
\item Case $x$: the equation is direct from the semantics which is independent 
of  $P$.  For $\p$-purity of $\means{D\proves x:t}\n P\h$,
the result holds by hypothesis on $\h$.
\item Case $\ifthenelse{e_1}{e_2}{e_3}$: straightforward use of induction.
\item Case $\LAM{x}{e}$: The equation holds because the semantics is independent of
  $P$.  Purity holds by induction on $e$.  
\item Case $e_1 e_2$: To show the equation, we use that $\means{D\proves e_1}$ is
  $\p$-pure, which holds by induction.  To show purity, we again use purity of
  $e_1$ as well as purity of $e_2$. 
\item Case $\letdecl{\recdecl{f}{x}{e_1}}{e_2}$:
By induction on $e_2$, using Lemma~\ref{lem:purlub}.
\item Case $\signs{\n'}{e}$:
The equation is direct from semantics, using the fact that 
$(P\intersect \A(\n'))-\{\p\} = (P-\{\p\})\intersect \A(\n')$.
\item Case $\enable{\p'}{e}$:  
We first consider the case where $\p'$ is distinct from $\p$.  We have 
\[\begin{array}{cll}
 & \means{\D\proves\enable{\p'}{e}:\ty}\n P \h \\
= & \means{\D\proves e:\ty}\n (P\sqcup_{\n} \{\p'\} ) \h& \mbox{semantics}\\
= & \means{\D\proves e:\ty}\n ((P\sqcup_{\n} \{\p'\} )-\{\p\}) \h & \mbox{induction hyp.}\\
= & \means{\D\proves e:\ty}\n ((P-\{\p\}\sqcup_{\n} \{\p'\} ) \h & \p,\p'\mbox{distinct}\\
= & \means{\D\proves\enable{\p'}{e}:\ty}\n (P-\{\p\}) \h& \mbox{semantics}
\end{array} \]
In case $\p'$ is $\p$ we have
\[\begin{array}{cll}
 & \means{\D\proves\enable{\p}{e}:\ty}\n P \h\\
= & \means{\D\proves e:\ty}\n (P\sqcup_{\n} \{\p\} ) \h& \mbox{semantics}\\
= & \means{\D\proves e:\ty}\n ((P-\{\p\}) \sqcup_{\n} \{\p\} ) \h&\mbox{see below}\\
= & \means{\D\proves\enable{\p}{e}:\ty}\n (P-\{\p\}) \h& \mbox{semantics}
\end{array} \]
The middle step is by cases on whether $\p\in\A(\n)$.  If it is, the
step holds by simply by definition of $\sqcup_{\n}$.  If not, the step
holds by induction on $e$.
\item Case $\chk{\p'}{e}$: Here $\p'$ is distinct from $\p$, by $\p$-purity.
We observe
\[\begin{array}{cll}
 & \means{\D\proves \chk{\p'}{e}:\ty}\n P \h\\
= & \mifthenelse{\p'\in P}{ \means{\D\proves e:\ty}\n P\h}{\star} &\mbox{  semantics}\\
= & \mifthenelse{\p'\in P-\{\p\}}{ \means{\D\proves e:\ty}\n (P-\{\p\}) \h}{\star} 
     &\mbox{  $\p',\p$ distinct, ind. for $e$}\\ 
= & \means{\D\proves \chk{\p'}{e}:\ty}\n (P-\{\p\}) \h
\end{array}
\]
\item Case $\test{\p'}{e_1}{e_2}$: Again, $\p'$ is distinct from $\p$,
and the argument is similar to $\CHK$.  
\end{itemize}
\end{xproof} % skip \blackslug since a theorem follows immediately

\begin{theorem}
\label{thm:elim}
For all $\n$, all $\p\in\A(\n)$, and all  $\p$-pure closed terms $e$
\[
\begin{array}{lcl}
\signs{\n}{ \enable{\p}{ \chk{\p}{e}}} 
& = & \signs{\n}{ e } 
\end{array}
\]  
\end{theorem}
\begin{proof}
Let $\h$ be the empty environment, for $e$ which is closed.
We observe for any $\n',P$:
\[\begin{array}{cll}
& \means{ \signs{\n}{ \enable{\p}{ \chk{\p}{e}}} }\n' P\h \\
=&\means{ \enable{\p}{ \chk{\p}{e}} }\n(\A(\n)\intersect P)\h & \mbox{ semantics}\\
=&\means{  \chk{\p}{e} }\n((\A(\n)\intersect P)\sqcup_{\n}\{\p\})\h & \mbox{ semantics}\\
=&\means{  \chk{\p}{e} }\n((\A(\n)\intersect P)\union\{\p\})\h & 
  \mbox{ def $\sqcup_{\n}$, using $\p\in\A(\n)$}\\ 
=&\means{  e }\n((\A(\n)\intersect P)\union\{\p\})\h & \mbox{ semantics}\\
=&\means{  e }\n(\A(\n)\intersect P)\h & \mbox{ $e$ and $\h$ are $\p$-pure,
  Lemma~\ref{lem:pur} } \\
=&\means{ \signs{\n}{e} }\n' P\h & \mbox{ semantics}
\end{array} \]
In the penultimate step, two uses are needed for the lemma: to remove
$\p$ and, in the case that $\p\in P$, to add it back.
\end{proof}

In order to deal with the password example we need the following conditional equivalence.% which can be used to remove a check.

\begin{theorem}
\label{thm:commute}
For all $\n$, all $\p\in\A(\n)$, and all terms $e$
\[
\begin{array}{lcl}
\signs{\n}{\chk{\p}{e}}
& = & \chk{\p}{\signs{\n}{e}}
\end{array}
\]  
\end{theorem}
\begin{proof}
We observe for any $\n',P, \h$:
\[\begin{array}{cll}
& \means{\signs{\n}{\chk{\p}{e}}}\n' P\h \\
=&\means{\chk{\p}{e}}\n(P\intersect \A(\n))\h & \mbox{ semantics}\\
=&\mifthenelse{\p\in(P\intersect\A(\n))}{\means{e}\n(P\intersect\A(\n))\h}{\star} &\mbox{ semantics}\\
=&\mifthenelse{\p\in P}{\means{e}\n(P\intersect\A(\n))\h}{\star} &\mbox{ sets, since $\p\in\A(\n)$}\\
=&\mifthenelse{\p\in P}{\means{\signs{\n}{e}}\n' P \h}{\star} &\mbox{ semantics}\\
=&\means{\chk{\p}{\signs{\n}{e}}}\n' P \h &\mbox{ semantics }%\; \blackslug}
\end{array} 
\]
\end{proof}
The above proofs are examples of the benefit of a compositional
semantics.  The proofs are by direct calculation, without need for
induction.  For Theorem~\ref{thm:elim}, the proof goes through for
open terms as well, if the environment $\h$ is pure.  One expects
built-in constants to have pure and safe values.

%%%%%%%%%%%%%%%%%%%%%%%%%%%%%%%%%%%%%%%%%%%%%%%%%%%%%%%%%%%%%%%%%%%%%%%%%
\subsection{The password example}
%%%%%%%%%%%%%%%%%%%%%%%%%%%%%%%%%%%%%%%%%%%%%%%%%%%%%%%%%%%%%%%%%%%%%%%%%

We now revisit the password example, using % (Section~\ref{sec:peg}),
Theorems~\ref{thm:elim} and~\ref{thm:commute} to eliminate checks.
We abbreviate $user,root$ as $u,r$.
\[\begin{array}{cll}
& passwd(\mbox{``mypass''}) \\
=& \{\mbox{because }passwd=(\LAM{x}{\signs{r}{\chk{p}{\enable{w}{writepass(x)}}}})\}\\
%(\LAM{x}{\signs{r}{\chk{p}{\enable{w}{writepass(x)}}}})(\mbox{``mypass''})  \\
& \signs{r}{\chk{p}{\enable{w}{writepass(\mbox{``mypass''})}}}  \\
=& \{\mbox{because }writepass=(\LAM{x}{\signs{r}{\chk{w}{hwWrite(x,\mbox{``/etc/password''})}}})\}\\
%& \signs{r}{\chk{p}{\enable{w}{
%    (\LAM{x}{\signs{r}{\chk{w}{hwWrite(x,\mbox{``/etc/password''})}}}){
%      (\mbox{``mypass''})}}}} \\ 
& \signs{r}{\chk{p}{\enable{w}{
    \signs{r}{\chk{w}{hwWrite(\mbox{``mypass''},\mbox{``/etc/password''})}}}}} \\
=&\{\mbox{by Theorem~\ref{thm:commute} since }\A(r)=\{p,w\}\}\\
 &\chk{p}{\signs{r}
                 {\enable{w}
                         {\chk{w}
                              {\signs{r}
                                     {hwWrite(\mbox{``mypass''},
                                              \mbox{``/etc/password''})}}}}}\\
=&\{\mbox{by Theorem~\ref{thm:elim} since $w\in\A(r)$ and }
    \signs{r}{hwWrite(\ldots)}
    \mbox{ is $\p$-pure closed}\}\\
&\chk{\p}{\signs{r}
           \signs{r}{hwWrite(\mbox{``mypass''},\mbox{``/etc/password''})}}\\
=&\{\mbox{because }\signs{\n}{\signs{\n}{e}}=\signs{\n}{e}\}\\
&\chk{\p}{\signs{r}{hwWrite(\mbox{``mypass''},\mbox{``/etc/password''})}}
\end{array}
\]
In the last step we used the unconditional equation $\signs{\n}{\signs{\n}{e}}=\signs{\n}{e}$ which is easily proved.
Finally, we obtain:
%=&\{\mbox{by Theorem~\ref{thm:commute} since }\p\in\A(r)\}\\
%&\signs{r}{\chk{\p}{hwWrite(\mbox{``mypass''},\mbox{``/etc/password''})}}
\[
\begin{array}{cll}
& use=\signs{u}{\enable{\p}{passwd(\mbox{``mypass''})}}\\
=&\signs{u}{\enable{\p}{\chk{\p}
                        \signs{r}{hwWrite(\mbox{``mypass''},
                                          \mbox{``/etc/password''})}}}\\
=&\{\mbox{by Theorem~\ref{thm:elim} since $\p\in\A(u)$ and }
    \signs{r}{hwWrite(\ldots)}
    \mbox{ is $\p$-pure closed}\}\\
&\signs{u}{\signs{r}{hwWrite(\mbox{``mypass''},\mbox{``/etc/password''})}}
\end{array}
\]
%\comment{Last couple of steps need to be checked.}

%%%%%%%%%%%%%%%%%%%%%%%%%%%%%%%%%%%%%%%%%%%%%%%%%%%%%%%%%%%%%%%%%%%%%%%%%
\subsection{Other transformations}
%%%%%%%%%%%%%%%%%%%%%%%%%%%%%%%%%%%%%%%%%%%%%%%%%%%%%%%%%%%%%%%%%%%%%%%%%
We now give an example of a transformation which employs a weaker notion of
equality than the ones in Section~\ref{sec:elimcheck}, cf.~(\ref{eq:dependDenot}).
This is the \textsf{TestGrant} equation of Fournet and Gordon,\footnote{Theirs is slightly more general, as their $\TEST$ and $\ENABLE$ apply to permission sets; this can be desugared to singleton permissions as in our notation.}
 which, in our notation amounts to proving,
for all $\n, P, h, p, e_1, e_2$, with $P \subseteq \A(\n)$,
\begin{equation}
\label{eq:tg}
\means{\test{p}{e_1}{e_2}}\n P h = \means{\test{p}{\enable{p}{e_1}}{e_2}}\n P h
\end{equation}
To show (\ref{eq:tg}) we use the denotational semantics and prove
\[
\begin{array}{lcl}
\mifthenelse{p\in P}{\means{e_1}\n P h}{\means{e_2}\n P h}
&=&
\mifthenelse{p\in P}{\means{\enable{p}{e_1}}\n P h}{\means{e_2}\n P h}
\end{array}
\]
It suffices to prove that when $p \in P$, we have
$
\means{e_1}\n P h = \means{\enable{p}{e_1}}\n P h
$.
We calculate:
\[
\begin{array}{cll}
& \means{\enable{p}{e_1}}\n P h\\
= & \means{e_1}\n (P \lub_n \{p\}) h & \mbox{ semantics}\\
%= & \{\mbox{because $p \in P$ and $P \subseteq \A(n)$ implies $p\in \A(n)$}\} \\
= & \means{e_1}\n P h & \mbox{ because $p \in P$ and $P \subseteq \A(n)$ implies $p\in \A(n)$}
\end{array}
\]

We have proved the correctness of several other primitive equations
in Fournet and Gordon's paper~\cite[Section 4.1]{FournetG03}.\footnote{For a few we need to consider the evident generalization of our language that allows permission sets in $\ENABLE$ (for \textsf{Grant Grant}, 
\textsf{Grant Frame}, and \textsf{Frame Grant Intersect}) and in $\TEST$ (for \textsf{Test $\union$}).
}
Specifically,
\textsf{Frame Frame, Frame Frame Frame, Frame Frame Grant, 
Frame Grant, 
Frame Grant Frame, 
Frame Grant Test, 
Frame Test Then, 
Grant Grant, % union in grant
Grant Frame, % intersection in grant
Grant Frame Grant}
 and \textsf{Test $\union$}. % union in test
We have also proved the correctness of their derived equations 
\textsf{Frame Appl, Frame Frame Intersect} and 
\textsf{Frame Grant Intersect}.  % intersection in grant
Our proofs of these equations do not require the restriction $P \subseteq \A(n)$.

%% \begin{description}
%% \item[(a)] (The framed expression) $R[v\;w]$ (recall Footnote~\ref{fn:pfg}) 
%% is contextually equivalent to $v\;w$ where
%% $v = \LAM{x}{S[e]}$ with $S \subseteq R$. Like $v$, the expression $w$ is a \emph{value}, that is, it is either the boolean constant $\TRUE$ or of the form $\LAM{x}{f}$ for some expression $f$. The subset condition merely states
%% that the callee has at most the permission as the caller.
%% \item[(b)] $R[\enable{T}{v\;w}]$ is contextually equivalent to $v'\;w$ where
%% $v' = \LAM{x}{S[\enable{T}{e}]}$ with $T \union S \subseteq R$. ($v$ is the same as in (a).)
%% \end{description}

%% In our language, define $v \eqdef \LAM{x}{\signs{\n_1}{e}}$ where $\A(\n_1) = S$. Then we can show, with $\A(\n_2) = R$, and letting $w$ be a standard expression
%% \[
%% \signs{\n_2}{v\;w} = v\;w
%% \]
%% The denotation of the LHS yields, for principal $\n$, permissions $P$ and environment $h$,
%% \[
%% \means{e}\n_1 (P \intersect S)[h \mid x \mapsto d]
%% \]
%% where $d = \means{w}\n_2 (P\intersect R) h$. Likewise the denotation of the RHS yields, for principal $\n$, permissions $P$ and environment $h$,
%% \[
%% \means{e}\n_1 (P \intersect S)[h \mid x \mapsto d_1]
%% \]
%% where $d_1 = \means{w}\n P h$. Observe that $w$ can either be $\TRUE$, whence
%% $d = d_1$; or it is a standard expression of the form $\LAM{x}{\signs{\n_3}{e}}$,
%% whence again $d = d_1$. Hence the denotations are equal.

%% Case (b) can be established \emph{mutatis mutandis}.

We have also proved the tail call elimination laws in~\cite[Section 5.2]{FournetG03}. The basic idea in tail call elimination is to not build a new frame for the last call of a function; instead the callee
can directly return to the caller's caller. 
Tail call elimination is problematic with stack inspection, as a stack frame holds the principal for the current code (or, equivalently, the principal's static permissions).
As noted earlier, Clements and Felleisen \cite{ClementsF04} give an abstract machine for which tail call elimination is sound and efficient.  The calculus and small-step semantics of Fournet and Gordon \cite[Section 5.2]{FournetG03} allows a limited modeling of tail call elimination.  Here is one of their two transformations, in our notation:
\begin{equation}\label{eq:tail}
 \signs{\n_2}{ ((\LAM{x}{ \signs{\n_1}{e_1}}) \; e_2) }
  = ((\LAM{x}{ \signs{\n_1}{e_1}}) \; e_2)
\end{equation}
From left to right this can be read as dropping the ``frame'' of the calling context.  In their setting it can be read as a transition, provided that $e_2$ is a value.
They show that this added transition is admissible, in the sense of not changing outcomes, provided that 
the callee's static permissions are among the caller's, \ie $\A(n_1)\subseteq\A(n_2)$. 
%\dn{And if there was much interest in this topic, we would also need to deal with letrec; but there isn't.}  
For our purposes, a value is a boolean literal, a variable (whose value is thus in the environment) or an abstraction.
We shall prove that the equation holds in our semantics, under these conditions.
For the proof, it is convenient to use the following easily proved fact which holds for any $D, e_1, e_2, \n, P, h$.
\begin{equation}\label{eq:appLam}
\means{D\proves (\LAM{x}{ e_1 }) \; e_2 }\n P \h
\;=\;
\means{D,x:t\proves  e_1 }\n P \ext{\h}{x}{ \means{D\proves e_2}\n P \h }
\end{equation}
Note: if $e_2$ is a value then its semantics $\means{D\proves e_2}\n P \h$ is independent from $P$ (though not necessarily from $\n$ or $\h$); see Fig.~\ref{fig:ds}.
To prove (\ref{eq:tail}) we observe
\[\begin{array}{clll}
  & \means{D\proves \signs{\n_2}{ ((\LAM{x}{ \signs{\n_1}{e_1}}) \; e_2) }}\n P \h \\
=  & \means{D\proves ((\LAM{x}{ \signs{\n_1}{e_1}}) \; e_2) }\n (P\intersect \A(\n_2)) \h 
      & \mbox{by semantics of signs} \\
=  & \means{D,x:t\proves \signs{\n_1}{e_1}}\n (P\intersect \A(\n_2)) 
                                         \ext{\h}{x}{ \means{D\proves e_2}\n (P\intersect \A(\n_2)) \h} 
      & \mbox{lemma (\ref{eq:appLam})}  \\
=  & \means{D,x:t\proves e_1}\n (P\intersect \A(\n_1)) 
                                          \ext{\h}{x}{ \means{D\proves e_2}\n (P\intersect \A(\n_2)) \h}
        & \mbox{sem., $\A(\n_1)\subseteq \A(\n_2)$} \\
=  & \means{D,x:t\proves e_1}\n (P\intersect \A(\n_1)) 
                                          \ext{\h}{x}{ \means{D\proves e_2}\n P \h}
         & \mbox{$e_2$ value, Note above} \\
=  & \means{D,x:t\proves \signs{n_1}{e_1}}\n P \ext{\h}{x}{ \means{D\proves e_2}\n P \h}
         & \mbox{semantics of signs} \\
=  & \means{D\proves ((\LAM{x}{ \signs{\n_1}{e_1}}) \; e_2)}\n P h
       & \mbox{lemma (\ref{eq:appLam})} 
\end{array}\]
%
%% % Same but less compact
%% \[\begin{array}{cll}
%%   & \means{D\proves \signs{\n_2}{ ((\LAM{x}{ \signs{\n_1}{e_1}}) \; e_2) }}\n P \h \\
%% = & \{\mbox{semantics of signs}\} \\
%%   & \means{D\proves ((\LAM{x}{ \signs{\n_1}{e_1}}) \; e_2) }}\n (P\intersect \A(\n_2)) \h \\
%% = & \{\mbox{lemma (\ref{eq:appLam})} \} \\
%%   & \means{D,x:t\proves \signs{\n_1}{e_1}}\n (P\intersect \A(\n_2)) 
%%                                          \ext{\h}{x}{ \means{D\proves e_2}\n (P\intersect \A(\n_2)) \h} \\
%% = & \{\mbox{semantics of signs, condition $\A(\n_1)\subseteq \A(\n_2)$}\} \\
%%   & \means{D,x:t\proves e_1}\n (P\intersect \A(\n_1)) 
%%                                           \ext{\h}{x}{ \means{D\proves e_2}\n (P\intersect \A(\n_2)) \h}
%% \\
%% = & \{\mbox{meaning of $e_2$ independent of permissions}\} \\
%%   & \means{D,x:t\proves e_1}\n (P\intersect \A(\n_1)) 
%%                                           \ext{\h}{x}{ \means{D\proves e_2}\n P \h}
%% \\
%% = & \{\mbox{semantics of signs}\} \\
%%   & \means{D,x:t\proves \signs{n_1}{e_1}}\n P \ext{\h}{x}{ \means{D\proves e_2}\n P \h}
%% \\
%% = & \{\mbox{lemma (\ref{eq:appLam})} \}\\
%%   & \means{D\proves ((\LAM{x}{ \signs{\n_1}{e_1}}) \; e_2)}\n P h
%% \end{array}\]
%
As with most of the other equations, the restriction $P\subseteq \A(\n)$ is not necessary here.
Fournet and Gordon give a second equation, also provable in our setting, that models tail calls involving $\ENABLE$.

%%%%%%%%%%%%%%%%%%%%%%%%%%%%%%%%%%%%%%%%%%%%%%%%%%%%%%%%%%%%%%%%%%%%%%%%%
\section{Using the Static Analysis}
\label{sec:use}
%%%%%%%%%%%%%%%%%%%%%%%%%%%%%%%%%%%%%%%%%%%%%%%%%%%%%%%%%%%%%%%%%%%%%%%%%
Section~\ref{sec:tran} gives several program transformations that can
be justified by the eager denotational semantics of our language. 
A more drastic transformation is possible under some conditions.
The safety results of Section~\ref{sec:sa}
show that if the static analysis derives a judgement 
$\Delta;\;\n\proves e:\th,\Pi$, then executing $e$ using a privilege set
that contains at least the enabled privileges $\Pi$ would not lead to a
security error. We should therefore be able to drop all $\ENABLE$'s and 
$\CHK$'s from $e$. If $e$ is $\TEST$-free, we can then show that 
the meaning of $e$ is the same as its meaning with $\ENABLE$'s and $\CHK$'s
erased. This is formalized below.
\begin{definition}
The erasure translation $(.)^-$ is defined as follows:
\[
\begin{array}{lcl}
\TRUE^- &=& \TRUE\\
x^-&=& x\\
(\ifthenelse{e_1}{e_2}{e_3})^- &=& \ifthenelse{e_1^-}{e_2^-}{e_3^-}\\
(\LAM{x}{e})^- &=& \LAM{x}{e^-} \\
(\letdecl{\recdecl{f}{x}{e_1}}{e_2})^- &=&
\letdecl{\recdecl{f}{x}{e_1^-}}{e_2^-}\\
(\signs{\n}{e})^- &=& \signs{\n}{e^-}\\
(\enable{\p}{e})^-&=& e^-\\
(\chk{\p}{e})^- &=& e^-\\
(\test{\p}{e_1}{e_2})^- & 
\multicolumn{2}{l}{\mbox{ is undefined}.}
\end{array}
\]
\end{definition}
\begin{theorem}\label{thm:erase}
Let $e$ be $\TEST$-free and let $\Empty;\;\n\proves e:\BOOL,\Ps$. Then for
all $P\in\P(\PRIV)$, if $\Ps\subseteq P$ then
$\means{\Empty\proves e:\BOOL}\n P\{\} = 
\means{\Empty\proves e^-:\BOOL}\n P\{\}$.
\end{theorem}
\begin{proof}
Immediate consequence of Lemma~\ref{lem:rem} and definition $\Rel~\BOOL$ below.
\end{proof}
\begin{definition}
For each annotated type $\th$ the relation $\Rel~\th$ on 
$\means{\th^*}_{\bot\star}$ is defined as follows:
For all $\th$, $\Rel~ {\th}~\bot~\bot$ holds and otherwise
$\Rel~\th~d~d'$ is false if either $d$ or $d'$ is in $\{\bot,\star\}$.
For values other than $\bot, \star$,
the definition is by induction on structure of $\th$.

\begin{eqnarray*}
\Rel~ {\BOOL}~b~b' &\iff & b=b' \\
\Rel~(\th_1\annoto\Ps\th_2)~f~f' &\iff &
\forall P\in\P(\PRIV) .\forall d,d'\in\means{\th_1^*}.\; \\
&&  \Ps\subseteq P \land 
\Rel~{\th_1}~d~d'\implies \Rel~ {\th_2}~(f P d)~(f' P d')
\end{eqnarray*}
For annotated type environment $\Delta$, the predicate  
$\Rel~\Delta$ on $\means{\Delta^*}$ is defined by
$\Rel~\Delta~\h~\h' \iff \dom(\h)=\dom(\h')\mbox{  and } 
\forall x\in\dom(\h) . \Rel~(\Delta.x)~(\h.x)~(\h'.x)$.
\end{definition}

\begin{factx}\label{fact:relMono}
  $\th\leq \th'$ and $\Rel~\th~d~d'$ imply $\Rel~\th'~d~d'$.
\end{factx}
\begin{proof}
By induction on derivation of $\th\leq\th'$.  The result is clear for
$\BOOL\leq\BOOL$. For $(\th_1\annoto{\Ps}\th_2)\leq(\th'_1\annoto{\Ps'}\th'_2)$,
assume $\Rel~(\th_1\annoto{\Ps}\th_2)~f~f'$. 
To show $\Rel~(\th'_1\annoto{\Ps'}\th'_2)~f~f'$, consider 
any $P\in\P(\PRIV)$, such that $\Ps'\subseteq P$, and any 
$d, d'\in\means{{\th'_1}^*}$ with $\Rel~\th'_1~d~d'$. From the subtyping, 
we know
that $\Ps\subseteq\Ps'$, hence $\Ps\subseteq P$. Moreover, by induction
on derivation of $\th'_1\leq\th_1$, we obtain $\Rel~\th'_1~d~d'$ implies
$\Rel~\th_1~d~d'$. Hence from assumption $\Rel~(\th_1\annoto{\Ps}\th_2)~f~f'$, 
we obtain $\Rel~\th_2(f P d)(f' P d')$. Now by induction on derivation 
$\th_2\leq\th'_2$, we obtain $\Rel~\th'_2(f P d)(f' P d')$.  
\end{proof}

\begin{factx}
\label{fac:reladm}
The relation $\Rel$ preserves lubs. That is, for any $\th$,
let $u, u':\nats\to\means{\th^*}_{\bot\star}$ be ascending chains. Then,
$\forall i.\Rel~\th~u_i~u'_i$ implies $\Rel~\th~(\Lub_{i}u_i)(\Lub_{i}u'_i)$.
\end{factx}
\begin{proof}
By structural induction on $\th$. When $\th=\BOOL$, we have
$\Lub_{i}u_i=\Lub_{i}u'_i= 
\mathsf{true}$ or $\mathsf{false}$ or $\bot$.  Thus the result holds by 
definition $\Rel$. \\
When $\th=(\th_1\annoto{\Ps}\th_2)$, assume $P\in\P(\PRIV)$ and 
$d, d'\in\means{\th_1^*}$, such that $\Ps\subseteq P$ and 
$\Rel~\th_{1}d~d'$. Then, from assumption $\Rel~(\th_1\annoto{\Ps}\th_2)~u_i~u'_i$ 
we obtain $\Rel~\th_2~(u_{i} P d)(u'_{i} P d')$ for every $i$. Hence, by 
the induction hypothesis on $\th_2$, we get 
$\Rel~\th_2~(\Lub_{i}(u_{i} P d))(\Lub_{i}(u'_{i} P d'))$. Because 
lubs are pointwise, we get $\Rel~\th_2~((\Lub_{i}u_{i}) P d)((\Lub_{i}u'_{i}) P d')$.
\end{proof}
\begin{lemma}
\label{lem:rem}
Suppose $\Delta;\;\n\proves e:\th,~\Ps$ is derivable and $e$ is $\TEST$-free.
Then for all $P\in\P(\PRIV)$, for all $\h, \h^-\in
\means{\Delta^*}$, if $\Rel~ {\Delta}~\h~\h^-$ and $\Ps\subseteq P$ then  
$\Rel~\th~u~u^-$, where $u=\means{\Delta^*\proves e:\th^*}\n P \h$ 
and $u^-=\means{\Delta^*\proves e^-:\th^*}\n P \h^-$. 
\end{lemma}
(Note that $h^-, u^-$ are just suggestively named identifiers whereas
$e^-$ is the erasure of $e$.)
Theorem \ref{thm:erase} follows from the lemma because $\Rel~\Empty~\{\}~\{\}$ and
by definition $\Rel~\BOOL$, $\means{\Empty\proves e:\BOOL}\n P\{\}=
\means{\Empty\proves e^-:\BOOL}\n P\{\}$.
\begin{xproof}
of Lemma. Go by induction on the typing derivation, 
$\Delta;\;\n\proves e:\th,~\Ps$.
Throughout, we assume $P\in\P(\PRIV)$ and $\h,\h^-\in\means{\Delta^*}$ and 
$\Rel~\Delta~\h~\h^-$. Let
$u=\means{\Delta^*\proves e:\th^*}\n P \h$ and 
$u^-=\means{\Delta^*\proves e^-:\th^*}\n P \h^-$ for each case of $e$.
\begin{itemize}
\item Case $\TRUE$:
Then, $u = \mathsf{true} = u^-$ and $\Rel~ {\BOOL}~u~u^-$ by definition $\Rel$.
%\\
%%%%%%%%%%%%%%%%%%%%%%%%%%%%%%%%%%%%%%%%%%%%%%%%%%%%%%%%%%%%%%%%%%%%%%%%
\item Case $x$:
Then, $u=\h.x$ and $u^-=\h^-.x$. And, $\Rel~\th~u~u^-$ follows
from assumption $\Rel~ {\Delta}~\h~\h^-$.
%\\
%%%%%%%%%%%%%%%%%%%%%%%%%%%%%%%%%%%%%%%%%%%%%%%%%%%%%%%%%%%%%%%%%%%%%%%%
\item Case $\ifthenelse{e}{e_1}{e_2}$:
Then $\Ps_1\union\Ps_2\union\Ps_3\subseteq P$, and 
\begin{eqnarray*}
u&=&\mifthenelse{b}{\means{\Delta^*\proves e_1:\th^*}\n P \h}
                   {\means{\Delta^*\proves e_2:\th^*}\n P \h}\\
u^-&=&\mifthenelse{b^-}{\means{\Delta^*\proves e_1^-:\th^*}\n P \h^-}
                       {\means{\Delta^*\proves e_2^-:\th^*}\n P \h^-}
\end{eqnarray*}
where 
$b=\means{\Delta^*\proves e:\BOOL}\n P \h$ and
$b^-=\means{\Delta^*\proves e^-:\BOOL}\n P \h^-$. By the induction hypothesis
on the typing derivation of $e$, noting that $\Ps_1\subseteq P$, we have 
$\Rel~\BOOL~b~b^-$.
If $b=\bot=b^-$ then $u=\bot=u^-$ and $\Rel~\th~\bot~\bot$.
Otherwise, $b=\mathsf{true}$ or $b=\mathsf{false}$. In the
former case, by the induction hypothesis on the typing derivation of 
$e_1$, noting that $\Ps_2\subseteq P$, we have 
$\Rel~ {\th}~u~u^-$. In the latter case, by the induction hypothesis on 
the typing derivation of $e_2$, noting that $\Ps_3\subseteq P$, we have 
$\Rel~ {\th}~u~u^-$. 
%\\
%%%%%%%%%%%%%%%%%%%%%%%%%%%%%%%%%%%%%%%%%%%%%%%%%%%%%%%%%%%%%%%%%%%%%%%%%
\item Case $\LAM{x}{e}$:
Then 
\(\begin{array}[t]{rcl}
u&=&
\lam{P'}
   {\lam{d}
        {\means{\Delta^*,x:\th_{1}^*\proves e:\th_{2}^*}\n P'\ext{\h}{x}{d}}}\\
u^-&=&
\lam{P'}
   {\lam{d^-}
        {\means{\Delta^*,x:\th_{1}^*\proves e^-:\th_{2}^*}\n P'\ext{\h^-}{x}{d^-}}}
\end{array}\)
\\
To prove $\Rel~(\th_{1}\annoto{\Ps}\th_{2})~u~u^-$, consider any 
$P'\in\P(\PRIV)$ and any $d,d^-\in\means{\th_{1}^*}$ such that 
$\Ps\subseteq P'$ and $\Rel~ {\th_{1}}~d~d^-$, to show 
$\Rel~\th_2~(u P' d')~(u^-P'd^-)$.
By semantics, 
\begin{eqnarray*}
u P' d &=& 
\means{\Delta^*,x:\th_{1}^*\proves e:\th_{2}^*}\n P'\ext{\h}{x}{d}\\
u^- P' d^- &=& 
\means{\Delta^*,x:\th_{1}^*\proves e^-:\th_{2}^*}\n P'\ext{\h^-}{x}{d^-}
\end{eqnarray*}
So the induction hypothesis for $e$ yields $\Rel~\th_2~(u P' d)~(u^- P' d^-)$
provided that $\Ps\subseteq P'$ and $\Rel~ {(\Delta,x:\th_{1})}\ext{\h}{x}{d}
\ext{\h^-}{x}{d^-}$.
We have $\Ps\subseteq P'$ by assumption, and 
$\Rel~ {(\Delta,x:\th_{1})}\ext{\h}{x}{d}\ext{\h^-}{x}{d^-}$ follows from
$\Rel~ {\Delta}~\h~\h^-$ and $\Rel~ {\th_{1}}~d~d^-$.
%\\
%%%%%%%%%%%%%%%%%%%%%%%%%%%%%%%%%%%%%%%%%%%%%%%%%%%%%%%%%%%%%%%%%%%%%%%
\item Case $e_1\;e_2$:
Let $f=\means{\Delta^*\proves e_1:(\th_{1}\annoto{\Ps}\th_{2})^*}\n P \h$ and 
$d=\means{\Delta^*\proves e_2:{\th'_1}^*}\n P \h$, so that $u=f P d$.  
Let $f^-=\means{\Delta^*\proves e_1^-:(\th_{1}\annoto{\Ps}\th_{2})^*}\n P \h^-$
and $d^-=\means{\Delta^*\proves e_2^-:{\th'_1}^*}\n P \h^-$, so that 
$u^-=f^- P d^-$.  
(Recall that $\th'_1\leq\th_1$ implies ${\th'_1}^*=\th_{1}^*$ so the
applications $f P d$ and $f^-Pd^-$ make sense.)  From $\Rel~\Delta~\h~\h^-$ 
and assumption $\Ps\union\Ps_1\union\Ps_2\subseteq P$,  
we get by induction on $e_1$ that $\Rel~(\th_{1}\annoto{\Ps}\th_{2})~f~f^-$, 
and we get $\Rel~ {\th'_1}~d~d^-$ by induction on $e_2$. 
By $\th'_1\leq \th_1$ and Fact~\ref{fact:relMono} we have $\Rel~
{\th'_1}~d~d^-\implies \Rel~ {\th_1}~d~d^-$. 
Then by definition $\Rel~(\th_{1}\annoto{\Ps}\th_{2})~f~f^-$, since
$\Ps\subseteq P$, we get $\Rel~ {\th_2}(f P d)(f^-P d^-)$.
%\\
%%%%%%%%%%%%%%%%%%%%%%%%%%%%%%%%%%%%%%%%%%%%%%%%%%%%%%%%%%%%%%%%%%%%%%%%%
\item Case $\letdecl{\recdecl{f}{x}{e_1}}{e_2}$:
Then, $\Ps\union\Ps_1\subseteq P$.\\
Now 
\(\begin{array}[t]{rcl}
u&=&\means{\Delta^*,f:(\th_1\annoto{\Ps}\th_2^)*\proves e_2:\th^*}\n P \ext{\h}{f}{\fix\;G}\\
u^-&=&\means{\Delta^*,f:(\th_1\annoto{\Ps}\th_2^)*\proves e_2^-:\th^*}\n P \ext{\h^-}{f}{\fix\;G^-}
\end{array}\)
\\
where
\(\begin{array}[t]{rcl}
G(g)&=&
\lam{P'}
    {\lam{d}
         {\means{\Delta^*,f:(\th_1\annoto{\Ps}\th_2)^*,x:\th_1^*
                 \proves e_1:\th_2^*}\n P'\recext{\h}{f}{g}{x}{d}}}\\
G^-(g^-)&=&
\lam{P'}
    {\lam{d^-}
         {\means{\Delta^*,f:(\th_1\annoto{\Ps}\th_2)^*,x:\th_1^*
                 \proves e_1^-:\th_2^*}\n P'\recext{\h^-}{f}{g^-}{x}{d^-}}}
\end{array}\)
\\
To show $\Rel~ {\th}~u~u^-$ by induction on $e_2$, we need 
$\Ps_1\subseteq P$ and 
\[
\Rel~ {(\Delta,f:\th_1\annoto{\Ps}\th_2)}~\ext{\h}{f}{\fix\;G}~\ext{\h^-}{f}{\fix\;G^-} \] 
The former follows from assumption $\Ps\union\Ps_1\subseteq P$.
The latter follows from assumption, $\Rel~ {\Delta}~\h~\h^-$, and 
$\Rel~(\th_1\annoto{\Ps}\th_2)(\fix\;G)(\fix\;G^-)$, which we now proceed
to show.

Now $\fix\;G=\Lub_{i}g_i$, where $g_0=\lam{P'}{\lam{d\in\means{\th_1^*}}{\bot}}$ and
$g_{i+1}=G(g_i)$. 
Also $\fix\;G^-=\Lub_{i}g_i^-$, where $g_0^-=\lam{P'}{\lam{d^-\in\th_1^*}{\bot}}$ and
$g_{i+1}^-=G^-(g_i^-)$. 
And, $\Rel~(\th_1\annoto{\Ps}\th_2)(\fix\;G)(\fix\;G^-)$ is a consequence 
of the following claim:
\begin{eqnarray}
  \label{eq:rel1}
& \forall i.\; \Rel~(\th_1\annoto{\Ps}\th_2)~g_i~g_i^- 
\end{eqnarray}
Then from Lemma~\ref{fac:reladm}, we get 
$\Rel~(\th_1\annoto{\Ps}\th_2)(\Lub_{i}g_i)(\Lub_{i}g_i^-)$.  
It remains to show (\ref{eq:rel1}), for which we proceed by induction on $i$.\\
Base case: Show $\Rel~(\th_1\annoto{\Ps}\th_2)~g_0~g_0^-$. Assume any 
$P'\in\P(\PRIV)$ and any $v, v^-\in\means{\th_1^*}$, such that 
$\Ps\subseteq P'$ and $\Rel~ {\th_1}~v~v^-$. Then 
$g_{0}P'v=\bot=g_{0}^-P'v^-$  
and $\Rel~ {\th_2}(g_{0}P'v)(g_{0}^-P'v^-)$. \\
Induction step: Assume $\Rel~(\th_1\annoto{\Ps}\th_2)~g_i~g_i^-$, to show
$\Rel~(\th_1\annoto{\Ps}\th_2)~g_{i+1}~g_{i+1}^-$. \\
Now
\( \begin{array}[t]{rcl}
g_{i+1}&=&
\lam{P'}{\lam{d}{\means{\Delta^*,f:(\th_1\annoto{\Ps}\th_2)^*,x:\th_1^*
\proves e_1:\th_2^*}\n P \recext{\h}{f}{g_i}{x}{d}}}\\
g_{i+1}^-&=&
\lam{P'}{\lam{d^-}{\means{\Delta^*,f:(\th_1\annoto{\Ps}\th_2)^*,x:\th_1^*
\proves e_1^-:\th_2^*}\n P \recext{\h^-}{f}{g_i^-}{x}{d^-}}} 
\end{array}
\)
\\
Assume any
$P'\in\P(\PRIV)$ and $v, v^-\in\means{\th_1^*}$, such that $\Ps\subseteq P'$
and $\Rel~ {\th_1}~v~v^-$.  
Then 
\begin{eqnarray*}
g_{i+1}P'v
&=&
\means{\Delta^*,f:(\th_1\annoto{\Ps}\th_2)^*,x:\th_1^*
\proves e_1:\th_2^*}\n P' \recext{\h}{f}{g_i}{x}{v}\\
g_{i+1}^-P'v^-
&=&
\means{\Delta^*,f:(\th_1\annoto{\Ps}\th_2)^*,x:\th_1^*
\proves e_1:\th_2^*}\n P' \recext{\h^-}{f}{g_i^-}{x}{v^-}
\end{eqnarray*}
Note that
$\Rel~(\Delta,f:\th_1\annoto{\Ps}\th_2,x:\th_1)~\recext{\h}{f}{g_i}{x}{v}~
\recext{\h^-}{f}{g_i^-}{x}{v^-}$.
Therefore, by the main induction hypothesis on the typing derivation 
$\Delta, f:\th_1\annoto{\Ps}\th_2, x:\th_1;\;\n
      \proves e_1:\th_2, ~\Ps$, since $\Ps\subseteq P$, we obtain
$\Rel~ {\th_2}(g_{i+1}P'v)(g_{i+1}^-P'v^-)$.
%\\
%%%%%%%%%%%%%%%%%%%%%%%%%%%%%%%%%%%%%%%%%%%%%%%%%%%%%%%%%%%%%%%%%%%%%%
\item Case $\signs{\n'}{e}$:
Then $\Ps\subseteq P$ and 
$u=\means{\Delta^*\proves e:\th^*}\n'(P\intersect\A(\n')) \h$. 
The induction hypothesis on the typing derivation
of $e$ can be used to obtain $\Rel~ {\th}~u~u^-$, 
because $\Ps\subseteq(P\intersect\A(\n'))$ which follows from 
assumption $\Ps\subseteq P$ and side condition $\Ps\subseteq \A(\n')$.
%\\
%%%%%%%%%%%%%%%%%%%%%%%%%%%%%%%%%%%%%%%%%%%%%%%%%%%%%%%%%%%%%%%%%%%%%%
\item Case $\enable{\p}{e}$:
Then $\Ps\subseteq P$ and 
$u=\means{\Delta^*\proves e:\th^*}\n(P\sqcup_{\n}\{\p\}) \h$.
By the induction hypothesis for $e$, noting that 
$(\Ps\sqcup_{\n}\{\p\})\subseteq(P\sqcup_{\n}\{\p\})$, we have 
$\Rel~\th~u~\means{\Delta^*\proves e^-:\th^*}\n(P\sqcup_{\n}\{\p\}) \h^-$.
But now $e^-$ is $\p$-pure. So by Lemma~\ref{lem:pur},
$\means{\Delta^*\proves e^-:\th^*}\n(P\sqcup_{\n}\{\p\}) \h^- =
\means{\Delta^*\proves e^-:\th^*}\n P\h^-$. 
But $u^-=\means{\Delta^*\proves e^-:\th^*}\n P\h^-$. 
Hence $\Rel~\th~u~u^-$.
%\\
%%%%%%%%%%%%%%%%%%%%%%%%%%%%%%%%%%%%%%%%%%%%%%%%%%%%%%%%%%%%%%%%%%%%%%%
\item Case $\chk{p}{e}$:
\begin{sloppypar}
Then $\Ps\union\{\p\}\subseteq P$, hence $\p\in P$.
Now \( u=\mifthenelse{\p\in P}{\means{\Delta^*\proves e:\th^*}\n P \h}{\star} \).
Since $\p\in P$, we have,
$u=\means{\Delta^*\proves e:\th^*}\n P \h$ and, by the induction
hypothesis on the typing derivation of $e$, we have 
$\Rel~\th~u~\means{\Delta^*\proves e^-:\th^*}\n P \h^-$. Hence
$\Rel~ {\th}~u~u^-$.
\end{sloppypar}
%%%%%%%%%%%%%%%%%%%%%%%%%%%%%%%%%%%%%%%%%%%%%%%%%%%%%%%%%%%%%%%%%%%%%%%
\end{itemize}
\end{xproof}
%%%%%%%%%%%%%%%%%%%%%%%%%%%%%%%%%%%%%%%%%%%%%%%%%%%%%%%%%%%%%%%%%%%%%%%%%
\section{Stack Semantics}
\label{sec:stk}
%%%%%%%%%%%%%%%%%%%%%%%%%%%%%%%%%%%%%%%%%%%%%%%%%%%%%%%%%%%%%%%%%%%%%%%%%

This section gives a formal semantics using stack inspection, and
shows that for standard expressions it coincides with the eager
semantics.  The connection is much more direct than that of Wallach,
Appel and Felten, so a complete detailed proof is not very lengthy.

Because the operations on the stack are in fact stack-like, it is
straightforward to give a denotational style semantics parameterized
on the stack.  We define $\stacks =
\mathsf{nonempty~list~of}(\PRINC\times\P(\PRIV) )$, taken as a cpo ordered by
equality.  The top is the head of the list, and we write infix $::$
for cons, so $\tuple{\n,P}::S$ is the stack with $\tuple{\n,P}$ on top
of $S$, as in Section~\ref{sec:ov}.  We also use the predicate
$\check$ defined there, and recall the definition $p\in \privs~S \iff
\check(p,S)$.

\begin{factx}
\label{fac:emp}
For all $S$ and all $n$ we have
$\privs(S)\intersect\A(\n)=\privs(\tuple{\n,\Empty}::S)$.
\end{factx}
\begin{xproof}
The sets are equal because for any $p$
\[\begin{array}{lcll}
\p\in\privs(\tuple{\n,\Empty}::S) 
&\iff& \check(\p,(\tuple{\n,\Empty}::S)) &\mbox{ by def $\privs$} \\
&\iff& \p\in\A(\n) \land \check(\p,S) 
  &\mbox{ by def $\check$ and $p\not\in\Empty$} \\
&\iff& \p\in\A(\n) \land \p\in\privs(S)&\mbox{ by def $\privs$ \; \blackslug} 
\end{array} \]
\end{xproof}
The stack semantics of an expression is a function
\[ \meanss{\D\proves e:\ty}\in \stacks
\to\meanss{\D}\to\meanss{\ty}_{\bot\star} \]
Just as in the eager semantics, we need to account for dynamic binding
of privileges by interpreting arrow types using an extra parameter.
The stack semantics of types is as follows.
\[
\begin{array}{lcl}
\meanss{\BOOL} &=& \{ \mathsf{true},  \mathsf{false} \} 
\\
\meanss{\ty_1\to\ty_2} &=& \stacks\to\meanss{\ty_1}\to\meanss{\ty_2}_{\bot\star}
\end{array}
\]
The semantics of expressions is in Figure~\ref{fig:ss}.

\begin{figure*}
\hrule
\medskip
\begin{small}
\[
\begin{array}{lcl}
\meanss{\D\proves\TRUE:\BOOL}S \h &=& \mathsf{true}
\\
\meanss{\D\proves x:\ty}S \h &=& h.x 
\\
\meanss{\D\proves\ifthenelse{e}{e_1}{e_2}:\ty}S \h &=&
  \vmlet{b\;=\;\meanss{\D\proves e:\BOOL}S \h}
       {\;\;\mifthenelse{b}
                    {\meanss{\D\proves e_1:\ty}S \h}
                    {\meanss{\D\proves e_2:\ty}S \h}}
\\
\meanss{\D\proves\LAM{x}{e}:\ty_1\to\ty_2}S \h &=&
\begin{array}[t]{l}
  \lam{S' \in \stacks }{
       \lam{d\in\meanss{\ty_1}}}{} \\
        \meanss{\D, x:\ty_1\proves e:\ty_2}S' \ext{\h}{x}{d}
\end{array}
\\
\meanss{\D\proves e_1\;e_2:\ty_2}S \h &=&
  \vmlet{f\;=\;\meanss{\D\proves e_1:\ty_1\to\ty_2}S \h}
        {\mletml{d}{\meanss{\D\proves e_2:\ty_1}S \h}{f S d}}
\\
\multicolumn{3}{l}{
\meanss{\D\proves\letdecl{\recdecl{f}{x}{e_1}}{e_2}:\ty}S \h} \\
\multicolumn{3}{l}{
\qquad 
= \vmlet{G(g)\;=\;
\lam{S'}
    {\lam{d}
         {\meanss{\D,f:\ty_1\to\ty_2,x:\ty_1\proves e_1:\ty_2}S
          \recext{\h}{f}{g}{x}{d}}}}
{\meanss{\D,f:\ty_1\to\ty_2\proves e_2:\ty}S \ext{\h}{f}{\fix\;G}}} 
\\
\meanss{\D\proves\signs{\n'}{e}:\ty}S \h &=&
\meanss{\D\proves e:\ty} (\tuple{\n',\Empty}::S) \h
\\
\meanss{\D\proves\enable{\p}{e}:\ty}(\tuple{\n,P}::S) \h 
&=&
\meanss{\D\proves e:\ty}(\tuple{\n,P\union\{\p\}}::S) \h
\\
\meanss{\D\proves\chk{\p}{e}:\ty}S \h &=&
\mifthenelse{\check(\p,S)}{\meanss{\D\proves e:\ty}S \h}{\star}
\\
\meanss{\D\proves\test{\p}{e_1}{e_2}:\ty}S \h &=&
\mifthenelse{\check(\p,S)}{\meanss{\D\proves e_1:\ty}S \h }
            {\meanss{\D\proves e_2:\ty}S \h }
\end{array}
\]
\end{small}

\medskip
\hrule
\medskip
\caption{Stack semantics}
\label{fig:ss}
\end{figure*}

We can now relate the denotational semantics of Figure~\ref{fig:ds} to the
stack semantics of Figure~\ref{fig:ss}.
\begin{theorem}[Consistency]%~\\
\label{thm:cons}
For any standard expression $e$ and 
stack $(\tuple{\n, P'}::S)$, we have 
\[ \means{\Empty\proves e:\BOOL}\n P \{\}~=~
\meanss{\Empty\proves e:\BOOL}(\tuple{\n, P'}::S) \{\} 
\qquad\mbox{where $P=\privs(\tuple{\n, P'}::S)$}.
\]
\end{theorem}
\begin{proof}
Immediate consequence of Lemma~\ref{lem:sim} and definition  
$\Sim~\BOOL$ below.
\end{proof}

As in the proof of safety, we need to generalize the result to allow
nonempty contexts.  We also consider expressions of arrow type, for
which a logical relation is needed.

\begin{definition} 
Define data-type indexed family 
$\Sim~\ty \subseteq \means{\ty}_{\bot\star}\times\meanss{\ty}_{\bot\star}$ as 
follows.  
For any $\ty$, $\Sim~\ty~d~d'$ is true if $d=d'$ and $d\in\{\bot,\star\}$;
it is false if $d\neq d'$ and $d$ or $d'$ is in $\{\bot,\star\}$.  Otherwise:
\begin{eqnarray*}
%\Sim~\ty~\bot~\bot &\iff & \mathsf{true}\\
%\Sim~\ty~\star~\star &\iff & \mathsf{true}\\
\Sim~ {\BOOL}~b~b' &\iff & b=b' \\
\Sim~(\ty_1\to\ty_2)~f~f' &\iff &
\forall S\in\stacks . \forall d\in\means{\ty_1} . \forall d'\in\meanss{\ty_1} . \\
&&  \Sim~\ty_1~d~d' \implies
\Sim~\ty_2~  (f~(\privs~S)~d) ~ (f'~S~d') 
\end{eqnarray*}
\end{definition}
\begin{sloppypar}
An environment $\h\in\means{\D}$ simulates an environment $\h'\in\meanss{\D}$, written 
$\Sim~\D~\h~\h'$, provided $\Sim~(\D.x)~(\h.x)~(\h'.x)$ for all 
$x\in\dom(\h)$.
\end{sloppypar}

\begin{lemma}
\label{lem:lub}
The relation $\Sim$ preserves lubs. That is, for any $\ty$,
if $u:\nats\to\means{\ty}$ and $u':\nats\to\meanss{\ty}$
are ascending chains and  
$\forall i.\Sim~\ty~u_i~u'_i$ then $\Sim~\ty~(\Lub_{i}u_i)~(\Lub_{i}u'_i)$.
\end{lemma}
\begin{proof}
Go by structural induction on $\ty$. Assume that $\Sim~\ty~u_i~u'_i$. When
$\ty=\BOOL$, by definition $\Sim$ we obtain, for each $i$, $u_i=u'_i$.
Thus $\Sim~\ty~(\Lub_{i}u_i)~(\Lub_{i}u'_i)$.

% TeX HACK - linebreaks rather than sloppypar, in order to get end-of-proof box.
%\begin{sloppypar}
When $\ty=\ty_1\to\ty_2$, consider any $P, S, d, d'$ with $P=\privs(S)$
and $\Sim~\ty_1~d~d'$. We must show 
\linebreak
$~\Sim~\ty_2~((\Lub_{i}u_i) P d)~((\Lub_{i}u'_i) S d')$, \ie, 
by definition of lubs we must show, 
$\Sim~\ty_2~\Lub_{i}(u_i P d)~\Lub_{i}(u'_i S d')$. 
\linebreak
By assumption, for every $i$, $\Sim~(\ty_1\to\ty_2)~u_i~u'_i$, hence, 
$\Sim~\ty_2~(u_i P d)~(u'_i S d')$ holds for each $i$. Therefore, 
\linebreak
by induction for $\ty_2$, we obtain 
$\Sim~\ty_2~\Lub_{i}(u_i P d)~\Lub_{i}(u'_i S d')$.
%\end{sloppypar}
\end{proof}

\begin{lemma}
\label{lem:sim}
For any  stack $(\tuple{\n,P'}::S)$, for any 
standard expression $e$, and any $\D,\ty,\h, \h'$,
let $u=\means{\D\proves e:\ty}\n P \h$ where $P=\privs(\tuple{\n,P'}::S)$, 
and let 
\( u'=\meanss{\D\proves e:\ty}(\tuple{\n,P'}::S)\h'\).
Then $\Sim~\D~\h~\h'\implies \Sim~\ty~u~u'$.
\end{lemma}
The Consistency Theorem follows from the lemma because 
$\Sim~\Empty~\{\}~\{\}$ and since $\Sim~\BOOL~u~u'$ implies $u=u'$.
%NOTE: This cannot possibly hold at higher types, because the semantics
%of arrows differs.  Simulation attempted below, but how to get at the
%principal needed in the lemma?  Need functions parameterized on
%principals after all?
\begin{xproof} of Lemma. Go by induction on $e$.  
\begin{itemize}
%\item \fbox{Case of $\TRUE$}~\\[2mm]
\item Cases $\TRUE$ and $x$:
Immediate from semantic definitions.
%\\
%%%%%%%%%%%%%%%%%%%%%%%%%%%%%%%%%%%%%%%%%%%%%%%%%%%%%%%%%%%%%%%%%%%%%%%%
%\item \fbox{Case of $x$}~\\[2mm]
%\\
%%%%%%%%%%%%%%%%%%%%%%%%%%%%%%%%%%%%%%%%%%%%%%%%%%%%%%%%%%%%%%%%%%%%%%%%
%\item \fbox{Case of $\ifthenelse{e}{e_1}{e_2}$}~\\[2mm]
\item Case $\ifthenelse{e}{e_1}{e_2}$:
%$\Rule{\D\proves e:\BOOL  \qquad
%       \D\proves e_1:\ty \qquad 
%       \D\proves e_2:\ty}
%      {\D\proves\ifthenelse{e}{e_1}{e_2}:\ty}$}~\\[2mm]
Directly by induction.
%\\
%%%%%%%%%%%%%%%%%%%%%%%%%%%%%%%%%%%%%%%%%%%%%%%%%%%%%%%%%%%%%%%%%%%%%%%%%
%\item \fbox{Case of $\LAM{x}{e}$}~\\[2mm]
\item Case $\LAM{x}{e}$:
%$\Rule{\D,x:\ty_1\proves e:\ty_2}
%      {\D\proves\LAM{x}{e}:\ty_1\to\ty_2}$}~\\[2mm]
Let $u=\means{\D\proves\LAM{x}{e}:\ty_1\to\ty_2}\n P \h$ and
let \[ u'=\meanss{\D\proves\LAM{x}{e}:\ty_1\to\ty_2}S \h' \] Then 
\(\begin{array}[t]{rcl}
u &=&
\lam{P'}
    {\lam{d}
    {\means{\D,x:\ty_{1}\proves e:\ty_{2}}\n P'\ext{\h}{x}{d}}}\\
u'&=&
\lam{S'}
    {\lam{d'}
         {\meanss{\D,x:\ty_{1}\proves e:\ty_{2}}S' \ext{\h'}{x}{d'}}}
\end{array}
\)\\
To show $\Sim~(\ty_1\to\ty_2)~u~u'$, need to show that for any $S'',
d'', d'''$, such that $\Sim~\ty_1~d''~d'''$, it is
the case that $\Sim~\ty_2~(u~(\privs~S'')~d'')~(u'~S''~d''')$.  By
standardness, $e$ is $\signs{\n'}{e'}$ for some $\n',e'$.  Thus we can
proceed as follows, using $e \equiv \signs{\n'}{e'}$ and semantics of
$\mathtt{signs}$.
\begin{eqnarray*}
u~(\privs~S'')~d''
&=&
\means{\D,x:\ty_1\proves e:\ty_2}\n~(\privs~S'')~\ext{\h}{x}{d''}\\
&=& 
\means{\D,x:\ty_1\proves e':\ty_2}\n'~(\privs(S'')\intersect\A(\n'))~\ext{\h}{x}{d''}\\[2mm]
u'~S''~d'''
&=&
\meanss{\D,x:\ty_1\proves e:\ty_2}~S''~\ext{\h'}{x}{d'''}\\
&=& 
\meanss{\D,x:\ty_1\proves e':\ty_2}(\tuple{\n',\Empty}::S'')~\ext{\h'}{x}{d'''}
\end{eqnarray*}
Note that by definition $\Sim$ and by assumption $\Sim~\ty_1~d''~d'''$,
we have, $\Sim~(\D,x:\ty_1)~\ext{h}{x}{d''}~\ext{h'}{x}{d'''}$. Furthermore,
by Fact~\ref{fac:emp}, 
$\privs(S'')\intersect\A(\n')=\privs(\tuple{\n',\Empty}::S'')$. 
Therefore, by induction for $e'$, we obtain, $\Sim~\ty_2~(u~(\privs~S'')~d'')~(u'~S''~d''')$. This is where we need Definition~\ref{def:std}.
%TBD Appears that we need a simulation to internalize the result at 
%arrow types.
%\\
%%%%%%%%%%%%%%%%%%%%%%%%%%%%%%%%%%%%%%%%%%%%%%%%%%%%%%%%%%%%%%%%%%%%%%%
%\item \fbox{Case of $e_1\;e_2$}~\\[2mm]
\item Case $e_1e_2$:
%$\Rule{\D\proves e_1: \ty_1\to\ty_2\qquad 
%       \D\proves e_2: \ty_1}
%      {\D\proves e_1\;e_2: \ty_2}$}
\(\begin{array}[t]{rcl}
\means{\D\proves e_1\;e_2:\ty_2}\n P \h &=&
  \vmlet{f\;=\;\means{\D\proves e_1:\ty_1\to\ty_2}\n P \h}
        {\mletml{d}{\means{\D\proves e_2:\ty_1}\n P \h}{f P d}}\\
\meanss{\D\proves e_1\;e_2:\ty_2}S \h' &=&
  \vmlet{f'\;=\;\meanss{\D\proves e_1:\ty_1\to\ty_2}S \h'}
        {\mletml{d'}{\meanss{\D\proves e_2:\ty_1}S \h}{f' S d'}}
\end{array}\)\\
%% \begin{eqnarray*} **)
%% \means{\D\proves e_1\;e_2:\ty_2}\n P \h &=& **)
%%   \vmlet{f\;=\;\means{\D\proves e_1:\ty_1\to\ty_2}\n P \h} **)
%%         {\mletml{d}{\means{\D\proves e_2:\ty_1}\n P \h}{f P d}}\\ **)
%% \meanss{\D\proves e_1\;e_2:\ty_2}S \h' &=& **)
%%   \vmlet{f'\;=\;\meanss{\D\proves e_1:\ty_1\to\ty_2}S \h'} **)
%%         {\mletml{d'}{\meanss{\D\proves e_2:\ty_1}S \h}{f' S d'}} **)
%% \end{eqnarray*} **)
Need to show $\Sim~\ty_2~(f P d)~(f' S d')$. Since $\Sim~\D~\h~\h'$ and
$P=\privs(S)$, therefore, by induction for $e_1$, we have
$\Sim~(\ty_1\to\ty_2)~f~f'$. Similarly, by induction for $e_2$, we have
$\Sim~\ty_1~d~d'$. Hence the result follows by definition $\Sim$  
since $P=\privs(S)$. This case of the proof shows the necessity of defining
the relation $\Sim$.
%\\
%%%%%%%%%%%%%%%%%%%%%%%%%%%%%%%%%%%%%%%%%%%%%%%%%%%%%%%%%%%%%%%%%%%%%%%%%
%\item \fbox{Case of $\letdecl{\recdecl{f}{x}{e_1}}{e_2}$}~\\[2mm]
\item Case $\letdecl{\recdecl{f}{x}{e_1}}{e_2}$:
%$\Rule{\D, f:\ty_1\to\ty_2, x:\ty_1
%       \proves e_1:\ty_2
%       \qquad
%       \D, f:\ty_1\to\ty_2\proves e_2:\ty}
%      {\D\proves\letdecl{\recdecl{f}{x}{e_1}}{e_2}: \ty}$}
\[
\begin{array}{lcl}
\multicolumn{3}{l}{
\means{\D\proves\letdecl{\recdecl{f}{x}{e_1}}{e_2}:\ty}\n P \h}\\
\multicolumn{3}{l}{
\qquad 
= \vmlet{G(g)\;=\;
\lam{P'}
    {\lam{d}
         {\means{\D,f:\ty_1\to\ty_2,x:\ty_1\proves e_1:\ty_2}\n P' 
          \recext{\h}{f}{g}{x}{d}}}}
{\means{\D,f:\ty_1\to\ty_2\proves e_2:\ty}\n P \ext{\h}{f}{\fix\;G}}
}\\
\multicolumn{3}{l}{
\meanss{\D\proves\letdecl{\recdecl{f}{x}{e_1}}{e_2}:\ty}S \h} \\
\multicolumn{3}{l}{
\qquad 
= \vmlet{G'(g')\;=\;
\lam{S'}
    {\lam{d'}
         {\meanss{\D,f:\ty_1\to\ty_2,x:\ty_1\proves e_1:\ty_2}S'
          \recext{\h'}{f}{g'}{x}{d'}}}}
{\meanss{\D,f:\ty_1\to\ty_2\proves e_2:\ty}S \ext{\h'}{f}{\fix\;G'}}} 
\end{array}
\]
%Note: using simulation, not equality...
%Let $f$ be the eager semantics of $e_1$, and let $f'$ be the lazy semantics of
%$e_1$.  Then $\Sim~(\d_1\to\d_2)~f~f'$ by induction.  ...TBD  
To show the result, it suffices to show 
$\Sim~(\ty_1\to\ty_2)~(\fix\;G)~(\fix\;G')$, because then we can use induction
for $e_2$, noting that 
$\Sim~(\D,f:\ty_1\to\ty_2)~\ext{\h}{f}{\fix\;G}~\ext{\h'}{f}{\fix\;G'}$,
and that $P=\privs(S)$.
Accordingly, we demonstrate the following claim:
\begin{eqnarray}
\label{sim:1}
& \forall i.\Sim~(\ty_1\to\ty_2)~g_i~g'_i 
\end{eqnarray}
Then from Lemma~\ref{lem:lub}, we get 
$\Sim~(\ty_1\to\ty_2)~\Lub_{i}g_i~\Lub_{i}g'_i$. This completes the
proof. To show (\ref{sim:1}), we proceed by induction on $i$. We have:
\begin{eqnarray*}
g_0 &=& 
\lam{P'}{\lam{d}{\bot}}\\
g_{i+1} &=&
\lam{P'}{\lam{d}
             {\means{\D,f:\ty_1\to\ty_2,x:\ty_1\proves e_1:\ty_2}\n P' 
              \recext{\h}{f}{g_i}{x}{d}}}\\
&=&\{\mathrm{because}\; e_1\equiv\signs{\n'}{e'_1}\;\mathrm{by\; standardness} \}\\
& &
\lam{P'}{\lam{d}
             {\means{\D,f:\ty_1\to\ty_2,x:\ty_1\proves e_1:\ty_2}
              \n' (P'\intersect\A(\n'))\recext{\h}{f}{g_i}{x}{d}}}\\[3mm]
g'_0 &=& 
\lam{S'}{\lam{d'}{\bot}}\\
g'_{i+1} &=&
\lam{S'}{\lam{d'}
             {\meanss{\D,f:\ty_1\to\ty_2,x:\ty_1\proves e_1:\ty_2} S'
              \recext{\h'}{f}{g'_i}{x}{d'}}}\\
&=&\{\mathrm{because}\; e_1\equiv\signs{\n'}{e'_1}\}\\
& &\lam{S'}{\lam{d'}
             {\meanss{\D,f:\ty_1\to\ty_2,x:\ty_1\proves e'_1:\ty_2} 
              (\tuple{\n',\Empty}::S')\recext{\h'}{f}{g'_i}{x}{d'}}}
\end{eqnarray*}
Clearly, $\Sim~(\ty_1\to\ty_2)~g_0~g'_0$, by
definition $\Sim$. To show $\Sim~(\ty_1\to\ty_2)~g_{i+1}~g'_{i+1}$, 
assume $\Sim~(\ty_1\to\ty_2)~g_{i}~g'_{i}$ (induction hypothesis), and that
for any $S'$ and $P'=\privs(S')$, $\Sim~\ty_1~d~d'$ holds.
Then \[
\Sim~(\D,f:\ty_1\to\ty_2,x:\ty_1)~\recext{\h}{f}{g_i}{x}{d}~\recext{\h'}{f}{g'_i}{x}{d'}
\]   
by definition $\Sim$ and since $\Sim~D~\h~\h'$. Now by Fact~\ref{fac:emp},
$P'\intersect\A(\n')=\privs(\tuple{\n',\Empty}::S')$, so by the main
induction hypothesis on $e'_1$, 
$\Sim~\ty_2~(g_{i+1}P'd)~(g'_{i+1}S'd')$ holds.
%\\
%%%%%%%%%%%%%%%%%%%%%%%%%%%%%%%%%%%%%%%%%%%%%%%%%%%%%%%%%%%%%%%%%%%%%
%\item \fbox{Case of $\signs{\n}{e}$}~\\[2mm]
\item 
\begin{sloppypar}
Case $\signs{\n}{e}$:
%$\Rule{\D\proves e:\ty}
%      {\D\proves \signs{\n}{e}:\ty}$}~\\[2mm]
We have:
$\means{\D\proves\signs{\n'}{e}:\ty}\n P \h =
\means{\D\proves e:\ty}\n'(P\intersect\A(\n')) \h$ and
$\meanss{\D\proves\signs{\n'}{e}:\ty}S \h' =
\meanss{\D\proves e:\ty} (\tuple{\n',\Empty}::S) \h'$
so the result holds by induction on
$e$ provided $P'\intersect\A(\n') = \privs(\tuple{\n',\Empty}::S)$.
But this equality holds by Fact~\ref{fac:emp}.
%\\
\end{sloppypar}
%%%%%%%%%%%%%%%%%%%%%%%%%%%%%%%%%%%%%%%%%%%%%%%%%%%%%%%%%%%%%%%%%%%%%%
%\item \fbox{Case of $\enable{\p}{e}$}~\\[2mm]
\item Case $\enable{\p}{e}$:
%$\Rule{\D\proves e:\ty}
%      {\D\proves \enable{\p}{e}:\ty}$}~\\[2mm] 
The result holds by induction for $e$, provided that
$P\sqcup_{\n}\{\p\} 
= \privs( \tuple{\n,P'\union\{\p\}}::S ) $.  This holds because for any $\p'$
\[\begin{array}{lcll}
& & \p'\in P\sqcup_{\n}\{\p\} \\
&\iff& \p'\in P \lor (\p'\in \A(\n) \land \p'=\p) & \mbox{by def $\sqcup_{\n}$} \\
&\iff& \check(\p', \tuple{\n,P'}::S ) 
       \lor (\p'\in \A(\n) \land \p'=\p) & \mbox{assumption, def $\privs$} \\
&\iff& (\p'\in \A(n)\land (\p'\in P'\lor \check(\p', S )) 
       \lor (\p'\in \A(\n) \land \p'=\p) & \mbox{def $\check$} \\
&\iff& \p'\in \A(n)\land (\p'\in P'\union\{\p\} \lor \check(\p', S ) ) & \mbox{logic
  and sets}\\
&\iff& \p'\in \privs( \tuple{\n,P'\union\{\p\}}::S ) & \mbox{defs $\check$ and
  $\privs$} 
\end{array} \]
%%%%%%%%%%%%%%%%%%%%%%%%%%%%%%%%%%%%%%%%%%%%%%%%%%%%%%%%%%%%%%%%%%%%%%%
%\item \fbox{Case of $\chk{\p}{e}$}~\\[2mm]
\item Case $\chk{\p}{e}$:
%$\Rule{\D\proves e:\ty}
%      {\D\proves \chk{\p}{e}:\ty}$}~\\[2mm]
Both semantics are conditional; the condition in one case is $\p\in P'$ and 
in the other case $\check(\p,S)$, and these are equivalent conditions by 
assumption $P'=\privs(S)$ for the Lemma.   In case the condition is true, 
the result holds by induction, which applies because for both semantics the 
security arguments for $e$ are unchanged.
If the condition is false, the result holds because both semantics are $\star$
and $\Sim~\ty~\star~\star$.
%\\
%%%%%%%%%%%%%%%%%%%%%%%%%%%%%%%%%%%%%%%%%%%%%%%%%%%%%%%%%%%%%%%%%%%%%%%
%\item \fbox{Case of $\test{\p}{e_1}{e_2}$}~\\[2mm]
\item Case $\test{\p}{e_1}{e_2}$:
%$\Rule{\D\proves e_1:\ty\qquad
%       \D\proves e_2:\ty}
%      {\D\proves \test{\p}{e_1}{e_2}:\ty}$}~\\[2mm]
Similar to the case for $\CHK$.
%%%%%%%%%%%%%%%%%%%%%%%%%%%%%%%%%%%%%%%%%%%%%%%%%%%%%%%%%%%%%%%%%%%%%%%%%%%%
\end{itemize}
\end{xproof}

%%%%%%%%%%%%%%%%%%%%%%%%%%%%%%%%%%%%%%%%%%%%%%%%%%%%%%%%%%%%%%%%%%%%%%%%
\section{Conclusion}
\label{sec:disc}
%%%%%%%%%%%%%%%%%%%%%%%%%%%%%%%%%%%%%%%%%%%%%%%%%%%%%%%%%%%%%%%%%%%%%%%%

Our work was motivated by the hope, inspired by discussions with Dave Schmidt, for more principled semantics of static analyses presented in the form of type and effect systems.  
Our work serves to demonstrate two attractive features of denotational semantics, which a decade ago seemed largely eclipsed by operational semantics.  
The first is proof of program equalities via equational reasoning on denotations.
The second is logical relations, defined by structural recursion on types.
%Although the pendulum has swung back since that time, the use of denotational semantics is found more in some parts of the research community than in the training of PhD students.
We are glad for the opportunity to demonstrate the utility of denotational semantics while celebrating the contributions of Dave Schmidt, so adept a practitioner of all forms of semantic modeling.

\bibliographystyle{eptcs}
\bibliography{schmidt.bib}

%% \begin{thebibliography}{10}

%% \bibitem{AbadiBLP93}
%% M.~Abadi, M.~Burrows, B.~Lampson and G.~Plotkin.
%% \newblock{A calculus for access control in distributed systems}.
%% \newblock{\emph{ACM Trans.\ Programming Languages and Systems}}
%% 15 (4), pp.~706-734, 1993.

%% \bibitem{ClementsF04}
%% J.~Clements and M.~Felleisen.
%% \newblock{A tail-recursive machine with stack inspection}.
%% \newblock{\emph{ACM Trans.\ Programming Languages and Systems}}
%% 26 (6), pp.~1029-1052, 2004.

%% \bibitem{FournetG03}
%% C.~Fournet and A.~D.~Gordon.
%% \newblock{Stack inspection: Theory and variants}.
%% \newblock{\emph{ACM Trans.\ Programming Languages and Systems}}
%% 25 (3), pp.~360-399, 2003.

%% \bibitem{Gong99}
%% Li Gong.
%% \newblock{\em Inside Java 2 Platform Security}.
%% \newblock{Addison-Wesley, 1999}.

%% \bibitem{PottierSS:ESOP01}
%% F.~Pottier, C.~Skalka and S.~Smith.
%% \newblock{A systematic approach to static access control}.
%% \newblock{\em Proceedings of ESOP}, pp.~30-45, 2001.

%% \bibitem{SkalkaS:ICFP00}
%% C.~Skalka and S.~Smith.
%% \newblock{Static enforcement of security with types}.
%% \newblock{\em Proceedings of the fifth ACM International Conference on Functional Programming}, pp.~34-45, 2000.

%% \bibitem{WallachAF00}
%% D.~Wallach, A.~Appel and E.~Felten.
%% \newblock{SAFKASI: a security mechanism for language-based systems}.
%% \newblock{\em ACM Trans.\ Software Engineering Methodology}
%% 9 (4), pp.~341-378, 2000.
%% \end{thebibliography}

%\tableofcontents

\end{document}